\documentclass[11pt,a4paper]{article}
\usepackage{graphicx}
\usepackage{ifpdf}
\usepackage{amsthm,amsfonts,amssymb,amsmath,euscript,comment,bbm,bm,mathabx}
\usepackage[utf8]{inputenc}
\usepackage[T1]{fontenc}
\usepackage{color}
\usepackage{textcomp}
\usepackage{slashed}
\usepackage[colorlinks=true,linkcolor=red,citecolor=blue]{hyperref}
\usepackage{mathrsfs}
\usepackage{dsfont}
\usepackage{geometry}
\usepackage{epsfig}
\usepackage{tikz}
\usetikzlibrary{arrows.meta}
\usetikzlibrary{3d}
\usepackage{float}
\usetikzlibrary{decorations.pathmorphing}
\usetikzlibrary{decorations.markings}
\usepackage{tikz-3dplot}
\usepackage{xcolor} 
\allowdisplaybreaks[3]

\def\WWW{\mathbb{W}}

\def\iso{\mathfrak{i}\mathfrak{s}\mathfrak{o}}
\def\so{\mathfrak{s}\mathfrak{o}}
\DeclareMathOperator{\sdiv}{div}
\newcommand{\W}{\mathbf{W}}

\DeclareMathOperator{\GL}{GL}
\newtheorem{thm}{Theorem}[section]
\newtheorem*{Aglu}{Annular Gluing Theorem}
\newtheorem*{Cglu}{Conic Gluing Theorem}

\newtheorem{prop}[thm]{Proposition}

\newtheorem{conj}[thm]{Conjecture}

\newtheorem{lem}[thm]{Lemma}
\newtheorem{df}[thm]{Definition}
\newtheorem{rk}[thm]{Remark}

\newcommand{\ep}{\varepsilon}

\newcommand{\ga}{\gamma}
\renewcommand{\c}{\cdot}

\DeclareMathOperator{\SO}{SO}
\DeclareMathOperator{\ISO}{ISO}
\def\e{\mathbf{e}}

\def\QQQ{\mathbb{Q}}
\def\MMM{\mathbb{M}}
\def\Lie{\mathcal{L}}
\def\Hd{\dot{\H}}
\def\ev{\mathbf{e}}

\def\hti{\widetilde{h}}
\def\pit{\widetilde{\pi}}
\def\SSS{{\mathbb{S}}}

\def\jk{\mathfrak{j}}
\DeclareMathOperator{\Orb}{Orb}
\renewcommand{\O}{{\mathbf{O}}}
\DeclareMathOperator{\Ric}{\bf Ric}

\DeclareMathOperator{\tr}{tr}

\renewcommand{\th}{\theta}

\renewcommand{\a}{\alpha}
\renewcommand{\b}{\beta}
\newcommand{\Ga}{\Gamma}
\newcommand{\pr}{\partial}

\newcommand{\g}{{\bf g}}

\def\cb{\mathbf{c}}
\newcommand{\x}{\mathbf{x}}

\newcommand{\R}{\mathbf{R}}
\numberwithin{equation}{section}
\def\La{\Lambda}
\def\xja{\langle x \rangle}
\def\rt{{\widetilde{r}}}

\def\ka{\kappa}
\def\XX{\mathcal{X}}
\def\les{\lesssim}
\def\NNN{\mathbb{N}}

\def\PPP{\mathbb{P}}

\def\pk{\mathfrak{p}}
\def\ck{\mathfrak{c}}

\def\ek{\mathfrak{e}}
\def\RRR{\mathbb{R}}

\def\th{\theta}

\newcommand{\mr}{\mathcal{R}}

\newcommand{\M}{\mathcal{M}}

\def\ep{\varepsilon}

\def\La{\Lambda}

\def\E{\mathbf{E}}
\def\P{\mathbf{P}}
\def\C{\mathbf{C}}
\def\J{\mathbf{J}}
\def\Q{\mathbf{Q}}

\def\D{\mathbf{D}}

\def\om{\omega}

\def\Om{\Omega}

\def\Si{\Sigma}
\def\si{\sigma}

\def\ga{\gamma}
\def\Ga{\Gamma}

\def\a{\alpha}
\def\b{\beta}

\def\de{\delta}
\def\De{\Delta}
\def\nab{\nabla}
\def\ov{\overline}

\def\les{\lesssim}
\def\RR{\mr}

\DeclareMathOperator{\supp}{supp}

\def\Ric{\mathbf{Ric}}

\def\etabf{\boldsymbol{\eta}}

\def\nbf{\mathbf{n}}
\def\ellbf{\boldsymbol{\ell}}

\def\G{\mathbf{G}}
\def\H{\mathbf{H}}
\def\X{\mathbf{X}}

\def\UUU{\mathbb{U}}

\def\refAglu{{\hypersetup{linkcolor=black}\hyperref[MOT1.7]{\textbf{Annular Gluing Theorem} }}}
\def\refCglu{{\hypersetup{linkcolor=black}\hyperref[Conicgluingthm]{\textbf{Conic Gluing Theorem} }}}
\title{Cauchy Data for Formation of Multiple Black Holes with Prescribed ADM Parameters}
\author{Dawei Shen and Jingbo Wan}
\begin{document}
\maketitle
\begin{abstract}
We give a simple construction of smooth, asymptotically flat vacuum initial data modeling a relativistic collapsing $N$--body system, with independently prescribed ADM energy, linear momentum, and angular momentum for each component, subject to the timelike condition $\E>|\P|$. The initial data contain no trapped surfaces, and the future development contains multiple causally independent trapped regions that dynamically form from localized subsets of the initial slice. In particular, the maximal development of data with well-separated collapsing components and relative motion is expected to yield spacetimes containing multiple black holes.
\end{abstract}
\section{Introduction} 
A central problem in mathematical general relativity is to understand which smooth, asymptotically flat Cauchy data for the Einstein vacuum equations lead to black hole formation. While stationary black holes such as the Kerr family are well studied, much less is known about the structure of initial data whose evolution produces black holes, particularly in configurations involving more than one black hole.

A fundamental breakthrough in this direction is due to Christodoulou \cite{Chr}, who showed that trapped surfaces can form dynamically from regular characteristic data. This method was further developed in \cite{An,AnLuk,Chen-K,KLR,kr}. In the Cauchy setting, Li and Yu \cite{LY} constructed smooth asymptotically flat vacuum initial data whose future development contains a trapped surface.
This was extended by Li and Mei \cite{LiMei} to a construction of vacuum spacetimes exhibiting black hole formation from Cauchy data. In our previous work \cite{ShenWan2}, we constructed smooth Cauchy initial data whose future development contains multiple causally independent trapped regions, without any initial trapped surfaces. See also \cite{GSW}, joint work with E.~Giorgi, for the construction of initial data for multiple collapsing (charged) boson stars. Each trapped region arises from a localized subset of the initial slice; however, the relativistic parameters of the collapsing components are not addressed. By construction, the initial data sets obtained in \cite{GSW,ShenWan2} consist of multiple mass concentrations with prescribed ADM energies and well-separated centers of mass. These data evolve, in finite time, into several well-separated $3$--dimensional black holes, with small linear and angular momenta. By analogy with Newtonian gravitation, such black holes are expected to merge into a single black hole. Hence, one does not expect the long-time evolution to exhibit multiple black holes.

The present paper addresses this limitation. Motivated by the conic gluing method introduced by Carlotto-Schoen in \cite{CS} (see also Mao-Tao \cite{MaoTao}), we introduce a simple framework in which each collapsing region is modeled on a Kerr sector with independently prescribed ADM energy, linear momentum, and angular momentum, subject to the timelike condition $\E>|\P|$. The initial data remain smooth and free of trapped surfaces, whereas their future development contains multiple dynamically forming trapped regions with controlled relativistic parameters.

Our approach exploits the diffeomorphism invariance of the Einstein vacuum equations and the resulting indeterminacy of the constraint equations. Conceptually, we treat general relativity as special relativity plus controlled lower-order corrections. Kerr initial data are placed in Kerr-Schild coordinates, and their ADM charges are organized in a form that transforms covariantly under the Poincar\'e group. This isolates the exact special relativistic transformation laws, while the remaining lower-order terms are controlled so that gluing theorems can be applied. Annular gluing replaces the Kerr black hole core by a short-pulse collapsing region, and conic gluing localizes and separates different boosted Kerr sectors. The resulting data may be viewed as a family of relativistic collapsing $N$--body initial data.

Let $(\Si,g)$ be a $3$--dimensional Riemannian manifold and let $k$ be a symmetric $2$--tensor on $\Si$. The Einstein vacuum constraint equations are
\begin{align}
\begin{split}\label{Econstraint}
R(g)+(\tr_g k)^2-|k|^2_g &= 0,\\
\sdiv_g (k-\tr_g k\, g) &= 0,
\end{split}
\end{align}
where $\nab$ denotes the Levi-Civita connection of $g$ and $R(g)$ its scalar curvature. By the local existence theorem of Choquet-Bruhat and Choquet-Bruhat-Geroch \cite{cb,cbg}, any solution $(\Si,g,k)$ of
\eqref{Econstraint} admits a unique global maximum hyperbolic development
$(\M,\g)$ solving
\begin{equation}\label{EVE}
\Ric(\g)_{\mu\nu}=0,
\end{equation}
in which $(\Si,g)$ embeds isometrically with the second fundamental form $k$.

Now we introduce the geometric notation used to localize collapsing regions. For $\om\in\mathbb S^2$, $0<\th<\frac{\pi}{2}$, and $y\in\RRR^3$, we define
\[
C_{\om,\th}(y):=\{x\in\mathbb R^3:\angle(x-y,\om)<\th\},
\]
which is the solid cone in $\RRR^3$ with center at $y$, center vector $\om$ and angle $\th$. We use the abbreviated notation $C_{\om,\th}:=C_{\om,\th}(0)$. We also denote $B_{r}(x)$ the ball center at $x$ with Euclidean radius $r$.

Our main result is the following theorem.
\begin{thm}\label{maintheorem}
Let $N\in\NNN$ and $s\ge3$.
For each $I=1,\dots,N$, prescribe parameters
\[
(\E_I,\P_I,\J_I)\in\RRR_+\times\RRR^3\times\RRR^3,\qquad\quad\E_I>|\P_I|,
\]
and choose $N$ pairwise disjoint cones $C_{\om_I,\th_I}$ with $\om_I\in\SSS^2$ and $0<\th_I<\frac{\pi}{2}$. Then there exist parameters
\[
(\de_I,R_I,\cb_I,y_I)\in\RRR_+\times\RRR_+\times\RRR^3\times\RRR^3,
\qquad I=1,\dots,N,
\]
and an initial data set $(\RRR^3,g,k)$ that solves the Einstein constraint equations \eqref{Econstraint}, such that the following hold:
\begin{enumerate}
\item \emph{Local structure}:
For each $I\in\{1,2,\ldots,N\}$,
\begin{align}\label{gkformula}
\begin{split}
(g,k)&=(e,0) \qquad\text{ in }\; B_{(1-2\de_I)R_I}(\cb_I),\\
(g,k)&=(g_I,k_I) \quad\text{ in }\; B_{32R_I}^c(\cb_I) \cap \bigl(C_{\om_I,\frac12\th_I}(y_I)\cup B_{\frac12}(y_I)\bigr),
\end{split}
\end{align}
where $(g_I,k_I)$ denotes the initial Kerr data centered at $\cb_I$ with ADM energy $\E_I$, linear momentum $\P_I$, and angular momentum $\J_I$.
See Figure \ref{fig:mainthm}.
\item \emph{Analytic control}: In the gluing region, we have
\begin{align}\label{gkest}
\begin{split}
R_I^{-1}\|g-e\|_{H^s(B_{64R_I}(\cb_I)\setminus\ov{B_{R_I}}(\cb_I))} +\|k\|_{H^{s-1}(B_{64R_I}(\cb_I)\setminus\ov{B_{R_I}}(\cb_I))}&\les 1,\\
\|(g-e,k)\|_{H_b^{s,\de}\times H_b^{s-1,\de+1}(\Om_I)}&\les 1,
\end{split}
\end{align}
where
\begin{align}\label{defOmI}
\Om_I :=\ov{\big(C_{\om_I,\th_I}(y_I)\cup B_1(y_I)\big)\setminus \big(C_{\om_I,\frac12\th_I}(y_I)\cup B_{\frac12}(y_I)\big)}.
\end{align}
\item \emph{Future trapped surfaces}: For each $I\in\{1,2,\ldots,N\}$, a trapped surface forms in the future domain of dependence $D^+(B_{R_I}(\cb_I))$.
\item \emph{No initial trapped surfaces}: $(\RRR^3,g,k)$ contains no trapped surfaces.
\end{enumerate}
\end{thm}
\begin{figure}[H]
\begin{center}
\begin{tikzpicture}[
line width=0.5pt,
>={Stealth[length=3pt,width=4pt]},
scale=0.75
]
\coordinate (O) at (0,0,0);
\draw[->, dashed] (xyz cs:x=-9.5) -- (xyz cs:x=9.5) node[right] {\scalebox{0.92}{$y$}};
\draw[->, dashed] (xyz cs:y=-9.1) -- (xyz cs:y=9.3) node[above] {\scalebox{0.92}{$z$}};
\draw[->, dashed] (xyz cs:z=-12) -- (xyz cs:z=12.8) node[above] {\scalebox{0.92}{$x$}};
\fill[green!20, opacity=0.3, even odd rule]
  (1.1,0)
    -- plot[variable=\t, domain=0:6.4] ({\t+1.1},{0.4*\t})   
    -- (7.5,2.56) arc (90:-90:0.8 and 2.56)                 
    -- plot[variable=\t, domain=6.4:0] ({\t+1.1},{-0.4*\t}) 
    -- cycle
  (1.1,0)
    -- plot[variable=\t, domain=0:6.4] ({\t+1.1},{0.2*\t})   
    -- (7.5,1.28) arc (90:-90:0.5 and 1.28)                 
    -- plot[variable=\t, domain=6.4:0] ({\t+1.1},{-0.2*\t}) 
    -- cycle;
\fill[blue!25, opacity=0.35]
  (1.1,0)
    -- plot[variable=\t, domain=0:6.4] ({\t+1.1},{0.2*\t})
    -- (7.5,1.28) arc (90:-90:0.5 and 1.28)
    -- plot[variable=\t, domain=6.4:0] ({\t+1.1},{-0.2*\t})
    -- cycle;
\draw plot[variable=\t, domain=0:6.4] ({\t+1.1},{0.2*\t});
\draw plot[variable=\t, domain=0:6.4] ({\t+1.1},{-0.2*\t});
\draw (7.5,1.28) arc (90:270:0.5 and 1.28);
\draw (7.5,-1.28) arc (-90:90:0.5 and 1.28);
\draw plot[variable=\t, domain=0:6.4] ({\t+1.1},{0.4*\t});
\draw plot[variable=\t, domain=0:6.4] ({\t+1.1},{-0.4*\t});
\draw (7.5,2.56) arc (90:270:0.8 and 2.56);
\draw (7.5,-2.56) arc (-90:90:0.8 and 2.56);
\shade[ball color = red!50, opacity = 0.9] (4.5,0) circle (0.45);
\node[below left] at (4.5,0) {\scalebox{0.98}{$\mathbf{E}_1$}};
\draw[->, thick] (4.5,0) -- (6.2,0.4) node[above] {\scalebox{0.98}{$\mathbf{P}_1$}};
\draw[dashed] (4.0,-2.0) -- (5.0,2.0) node[above] {\scalebox{0.98}{$\mathbf{J}_1$}};
\begin{scope}[rotate around={-12:(5.3,1.1)}] 
\draw[thick,->]
  (5.3,1.5)
    arc[
      start angle=0,
      end angle=310,
      x radius=0.45,
      y radius=0.2
    ];
\end{scope}
\fill[green!20, opacity=0.3, even odd rule]
  (-1.1,0)
    -- plot[variable=\t, domain=0:6.4] ({-\t-1.1},{0.4*\t})     
    -- (-7.5,2.56) arc (90:270:0.8 and 2.56)                    
    -- plot[variable=\t, domain=6.4:0] ({-\t-1.1},{-0.4*\t})     
    -- cycle
  (-1.1,0)
    -- plot[variable=\t, domain=0:6.4] ({-\t-1.1},{0.2*\t})     
    -- (-7.5,1.28) arc (90:270:0.5 and 1.28)                    
    -- plot[variable=\t, domain=6.4:0] ({-\t-1.1},{-0.2*\t})     
    -- cycle;
\fill[blue!25, opacity=0.35]
  (-1.1,0)
    -- plot[variable=\t, domain=0:6.4] ({-\t-1.1},{0.2*\t})
    -- (-7.5,1.28) arc (90:270:0.5 and 1.28)
    -- plot[variable=\t, domain=6.4:0] ({-\t-1.1},{-0.2*\t})
    -- cycle;    
\draw plot[variable=\t, domain=0:6.4] ({-\t-1.1},{0.2*\t});
\draw plot[variable=\t, domain=0:6.4] ({-\t-1.1},{-0.2*\t});
\draw (-7.5,1.28) arc (90:270:0.5 and 1.28);
\draw[dashed] (-7.5,-1.28) arc (-90:90:0.5 and 1.28);
\draw plot[variable=\t, domain=0:6.4] ({-\t-1.1},{0.4*\t});
\draw plot[variable=\t, domain=0:6.4] ({-\t-1.1},{-0.4*\t});
\draw (-7.5,2.56) arc (90:270:0.8 and 2.56);
\draw[dashed] (-7.5,-2.56) arc (-90:90:0.8 and 2.56);
\shade[ball color = red!70, opacity = 0.8] (-3,0) circle (0.25);
\node[below right] at (-3,0) {\scalebox{0.98}{$\mathbf{E}_2$}};
\draw[->, thick] (-3,0) -- (-3.9,0.4) node[above] {\scalebox{0.98}{$\mathbf{P}_2$}};
\draw[dashed] (-3.6,-1.0) -- (-2.4,1.0) node[above] {\scalebox{0.98}{$\mathbf{J}_2$}};
\begin{scope}[rotate around={150:(-2.5,0.55)}]
\draw[thick,->]
  (-2,0.55)
    arc[
      start angle=0,
      end angle=310,
      x radius=0.3,
      y radius=0.2
    ];
\end{scope}
\fill[green!20, opacity=0.3, even odd rule]
  (0,1.1)
    -- plot[variable=\t, domain=0:6.4] ({0.4*\t},{\t+1.1})    
    -- (2.56,7.5) arc (0:180:2.56 and 0.8)                    
    -- plot[variable=\t, domain=6.4:0] ({-0.4*\t},{\t+1.1})   
    -- cycle
  (0,1.1)
    -- plot[variable=\t, domain=0:6.4] ({0.2*\t},{\t+1.1})    
    -- (1.28,7.5) arc (0:180:1.28 and 0.5)                    
    -- plot[variable=\t, domain=6.4:0] ({-0.2*\t},{\t+1.1})   
    -- cycle;
\fill[blue!25, opacity=0.35]
  (0,1.1)
    -- plot[variable=\t, domain=0:6.4] ({0.2*\t},{\t+1.1})
    -- (1.28,7.5) arc (0:180:1.28 and 0.5)
    -- plot[variable=\t, domain=6.4:0] ({-0.2*\t},{\t+1.1})
    -- cycle;    
\draw plot[variable=\t, domain=0:6.4] ({0.2*\t},{\t+1.1});
\draw plot[variable=\t, domain=0:6.4] ({-0.2*\t},{\t+1.1});
\draw (-1.28,7.5) arc (180:360:1.28 and 0.5);
\draw (1.28,7.5) arc (0:180:1.28 and 0.5);
\draw plot[variable=\t, domain=0:6.4] ({0.4*\t},{\t+1.1});
\draw plot[variable=\t, domain=0:6.4] ({-0.4*\t},{\t+1.1});
\draw (-2.56,7.5) arc (180:360:2.56 and 0.8);
\draw (2.56,7.5) arc (0:180:2.56 and 0.8);
\shade[ball color = red, opacity = 0.8] (0,4.5) circle (0.35);
\node[above left] at (0,4.5) {\scalebox{0.98}{$\mathbf{E}_3$}};
\draw[->, thick] (0,4.5) -- (0.5,3.7) node[right] {\scalebox{0.98}{$\mathbf{P}_3$}};
\draw[dashed] (-0.8,3.0) -- (0.8,6.0) node[above] {\scalebox{0.98}{$\mathbf{J}_3$}};
\begin{scope}[rotate around={-30:(0.9,5.5)}] 
\draw[thick,->]
  (0.9,5.5)
    arc[
      start angle=350,
      end angle=20,
      x radius=0.35,
      y radius=0.17
    ];
\end{scope}
\fill[green!20, opacity=0.3, even odd rule]
  (0,-1.1)
    -- plot[variable=\t, domain=0:6.4] ({0.4*\t},{-\t-1.1})    
    -- (2.56,-7.5) arc (0:-180:2.56 and 0.8)                   
    -- plot[variable=\t, domain=6.4:0] ({-0.4*\t},{-\t-1.1})   
    -- cycle
  (0,-1.1)
    -- plot[variable=\t, domain=0:6.4] ({0.2*\t},{-\t-1.1})    
    -- (1.28,-7.5) arc (0:-180:1.28 and 0.5)                   
    -- plot[variable=\t, domain=6.4:0] ({-0.2*\t},{-\t-1.1})   
    -- cycle;
\fill[blue!25, opacity=0.35]
  (0,-1.1)
    -- plot[variable=\t, domain=0:6.4] ({0.2*\t},{-\t-1.1})
    -- (1.28,-7.5) arc (0:-180:1.28 and 0.5)
    -- plot[variable=\t, domain=6.4:0] ({-0.2*\t},{-\t-1.1})
    -- cycle;
\draw plot[variable=\t, domain=0:6.4] ({0.2*\t},{-\t-1.1});
\draw plot[variable=\t, domain=0:6.4] ({-0.2*\t},{-\t-1.1});
\draw (-1.28,-7.5) arc (180:360:1.28 and 0.5);
\draw[dashed] (1.28,-7.5) arc (0:180:1.28 and 0.5);

\draw plot[variable=\t, domain=0:6.4] ({0.4*\t},{-\t-1.1});
\draw plot[variable=\t, domain=0:6.4] ({-0.4*\t},{-\t-1.1});
\draw (-2.56,-7.5) arc (180:360:2.56 and 0.8);
\draw[dashed] (2.56,-7.5) arc (0:180:2.56 and 0.8);
\shade[ball color = red, opacity = 0.9] (0,-6) circle (0.45);
\node[above right] at (0,-6) {\scalebox{0.98}{$\mathbf{E}_4$}};
\draw[->, thick] (0,-6) -- (-1.1,-7.9) node[left] {\scalebox{0.98}{$\mathbf{P}_4$}};
\draw[dashed] (1.0,-7.8) -- (-1.0,-4.2) node[above] {\scalebox{0.98}{$\mathbf{J}_4$}};
\begin{scope}[rotate around={28:(-0.1,-4.8)}]
\draw[thick,->]
  (-0.05,-4.3)
    arc[
      start angle=-20,
      end angle=320,
      x radius=0.6,
      y radius=0.2
    ];
\end{scope}
\draw[dashed] (0,0) circle (1.1);
\draw[thin] (-1.1,0) arc (180:360:1.1 and 0.45);
\draw[dashed] (1.1,0) arc (0:180:1.1 and 0.45);
\shade[ball color=gray!10, opacity=0.2] (0,0) circle (1.1);
\draw[thin] (0,0) circle (1.1);
\draw[dashed] (1.1,0) circle (0.21);
\draw[thin] (1.1-0.21,0) arc (180:360:0.21 and 0.105);
\draw[dashed] (1.1+0.21,0) arc (0:180:0.21 and 0.105);
\shade[ball color=gray!10, opacity=0.2] (1.1,0) circle (0.21);
\draw[thin] (1.1,0) circle (0.21);
\draw[dashed] (1.1,0) circle (0.43);
\draw[thin] (1.1-0.43,0) arc (180:360:0.43 and 0.29);
\draw[dashed] (1.1+0.43,0) arc (0:180:0.43 and 0.29);
\shade[ball color=gray!10, opacity=0.2] (1.1,0) circle (0.43);
\draw[thin] (1.1,0) circle (0.43);
\draw[dashed] (-1.1,0) circle (0.21);
\draw[thin] (-1.1-0.21,0) arc (180:360:0.21 and 0.105);
\draw[dashed] (-1.1+0.21,0) arc (0:180:0.21 and 0.105);
\shade[ball color=gray!10, opacity=0.2] (-1.1,0) circle (0.21);
\draw[thin] (-1.1,0) circle (0.21);
\draw[dashed] (-1.1,0) circle (0.43);
\draw[thin] (-1.1-0.43,0) arc (180:360:0.43 and 0.29);
\draw[dashed] (-1.1+0.43,0) arc (0:180:0.43 and 0.29);
\shade[ball color=gray!10, opacity=0.2] (-1.1,0) circle (0.43);
\draw[thin] (-1.1,0) circle (0.43);
\draw[dashed] (0,1.1) circle (0.21);
\draw[thin] (0-0.21,1.1) arc (180:360:0.21 and 0.105);
\draw[dashed] (0+0.21,1.1) arc (0:180:0.21 and 0.105);
\shade[ball color=gray!10, opacity=0.2] (0,1.1) circle (0.21);
\draw[thin] (0,1.1) circle (0.21);
\draw[dashed] (0,1.1) circle (0.43);
\draw[thin] (0-0.43,1.1) arc (180:360:0.43 and 0.29);
\draw[dashed] (0+0.43,1.1) arc (0:180:0.43 and 0.29);
\shade[ball color=gray!10, opacity=0.2] (0,1.1) circle (0.43);
\draw[thin] (0,1.1) circle (0.43);
\draw[dashed] (0,-1.1) circle (0.21);
\draw[thin] (0-0.21,-1.1) arc (180:360:0.21 and 0.105);
\draw[dashed] (0+0.21,-1.1) arc (0:180:0.21 and 0.105);
\shade[ball color=gray!10, opacity=0.2] (0,-1.1) circle (0.21);
\draw[thin] (0,-1.1) circle (0.21);
\draw[dashed] (0,-1.1) circle (0.43);
\draw[thin] (0-0.43,-1.1) arc (180:360:0.43 and 0.29);
\draw[dashed] (0+0.43,-1.1) arc (0:180:0.43 and 0.29);
\shade[ball color=gray!10, opacity=0.2] (0,-1.1) circle (0.43);
\draw[thin] (0,-1.1) circle (0.43);
\node[left] at (6.2,6.1) {\scalebox{0.98}{Euclidean}};
\draw[->, green!60!black, line width=0.3pt]
  (8,3.3) to[out=-165,in=45] (7.5,1.5);
\node[above, text=green!60!black]
  at (8.1,3.3) {\scalebox{0.98}{conic gluing region}};
\draw[->, blue!60!, line width=0.3pt]
  (7.8,-3.6) to[out=165,in=-45] (7.4,-0.5);
\node[below, text=blue!60!]
  at (7.9,-3.6) {\scalebox{0.98}{boosted Kerr data}};  
\draw[->, red!70!black, line width=0.3pt]
  (4.1,-2.4) to[out=45,in=-135] (4.5,-0.2);
\node[below, text=red!70!black]
  at (4.1,-2.4) {\scalebox{0.98}{short pulse core}}; 
\end{tikzpicture}
\end{center}
\caption{An illustration for Theorem \ref{maintheorem}. Each collapsing component is supported in a disjoint conic sector $C_{\om_I,\th_I}(y_I)$. The data are exactly Euclidean inside $B_{(1-2\de_I)R_I}(\cb_I)$, coincide with a boosted Kerr initial data set with prescribed ADM parameters $(\E_I,\P_I,\J_I)$ in each conic region $B_{32R_I}^c(\cb_I)\cap\big(C_{\om_I,\frac12\th_I}(y_I)\cup B_{\frac12}(y_I)\big)$. The innermost short-pulse core replaces the Kerr interior and gives rise to a trapped surface in the future domain $D^+(B_{R_I}(\cb_I))$.}
\label{fig:mainthm}
\end{figure}
From the perspective of the \emph{final state conjecture}, Theorem \ref{maintheorem} may be viewed as a proposal for admissible multi-component collapsing initial data configurations in vacuum general relativity. A central ingredient of the final state picture is the nonlinear stability of the Schwarzschild and Kerr families of black hole spacetimes \cite{DHRT,GKS,KS,KS:Kerr1,KS:Kerr2,KS:main,Shen}. We refer to \cite{KSurvey} for a detailed discussion of these developments and their role in the final state conjecture.

Existing constructions of multi-black-hole initial configurations based on gluing and related methods (e.g.\ \cite{CCI,CM,Hintz}) provide a complementary class of examples, but they address settings in which black holes are already present on the initial slice. By contrast, the present work focuses on the dynamical formation of multiple black holes from completely regular Cauchy data free of trapped surfaces.

More precisely, the initial data constructed here contain no trapped surfaces and no black hole regions initially, while each collapsing component is arranged so that a trapped surface forms in its future domain of dependence. At the same time, the collapsing regions are equipped with independently prescribed ADM energy, linear momentum, and angular momentum parameters, subject only to the timelike condition $\E>|\P|$. In this sense, Theorem \ref{maintheorem} produces a relativistic collapsing $N$--body family: each component behaves, at the level of conserved quantities, like a massive spinning particle in special relativity, but the data evolve according to the fully nonlinear Einstein vacuum equations.

A basic question is whether the maximal future development of such data can contain multiple black holes with distinct asymptotic parameters, rather than merging into a single black hole. We do not study the long-time evolution here; however, the interpretation of $N$--body suggests a concrete two-body dichotomy in the weak interaction regime. When two components are well separated and their masses are small relative to the separation scale, their motion may be approximated by special relativistic kinematics, with interactions modeled by a Newtonian potential. This hybrid description yields an explicit escape threshold in terms of the conserved energy-momentum pairs and the initial separation, which we record below in a one-dimensional setting.

Let $(\E_1,\P_1)$ and $(\E_2,\P_2)$ satisfy $\E_i>|\P_i|$ for $i=1,2$ and we denote
\[
m_i:=\sqrt{\E_i^2-|\P_i|^2}.
\]
Assume a one-dimensional motion along the separation axis with opposite directions and initial separation $d_{12}>0$. We introduce the hybrid total energy
\[
\E_{\mathrm{tot}}(d):=\E_1(d)+\E_2(d)-\frac{m_1m_2}{d},
\]
where $\E_i(d)$ denotes the relativistic energy of the $i$--th body at separation $d$. Conservation of $\E_{\mathrm{tot}}$ between $d=d_{12}$ and $d=\infty$ yields
\begin{equation*}
\E_1+\E_2-\frac{m_1m_2}{d_{12}}=\E_1(\infty)+\E_2(\infty).
\end{equation*}
The threshold between escape and merger corresponds to the situation in which each body has a nonnegative kinetic energy at infinity, that is, $E_i(\infty)\ge m_i$ for $i=1,2$. Then, the escape condition takes the following form:\footnote{In the nonrelativistic regime $|P_i|\ll m_i$, this reduces to the Newtonian criterion
\[
\frac{|\P_1|^2}{2m_1}+\frac{|\P_2|^2}{2m_2}\geq\frac{m_1m_2}{d_{12}},
\]
so \eqref{eq:intro-escape} provides the expected relativistic refinement of the classical escape condition.}
\begin{equation}\label{eq:intro-escape}
\E_1+\E_2-m_1-m_2 \ge \frac{m_1 m_2}{d_{12}} .
\end{equation}

Motivated by the explicit escape condition \eqref{eq:intro-escape}, and by the expectation that subextremal Kerr spacetimes describe the dynamically stable vacuum black hole end states, we formulate the following conjectural two-body escape/merger dichotomy for the maximal future development of the initial data produced by Theorem \ref{maintheorem}.
\begin{conj}[Two-body escape/merger threshold]\label{conj:two-body-escape}
Fix
\[
(\E_1,\P_1,\J_1),\ (\E_2,\P_2,\J_2)\in\RRR_+\times\RRR^3\times\RRR^3,
\qquad
\E_i>|\P_i|,
\]
and set $m_i:=\sqrt{\E_i^2-|\P_i|^2}$. Assume in addition the subextremality condition
\[
|\J_i|<m_i^2,\qquad i=1,2.
\]
Assume the linear momenta are collinear and oppositely directed. Let $d_{12}>0$ denote the Euclidean separation between the two centers in Theorem \ref{maintheorem}, and let $(\RRR^3,g,k)$ be the corresponding initial data. Then, there exists a universal constant $\ka>0$ such that, for $d_{12}$ large compared to the gluing scales and the short-pulse parameters, the following holds:
\begin{enumerate}
\item[\emph{(i)}] \emph{Escape.}
If
\[
d_{12}\;\geq\;\ka\,\frac{m_1m_2}{(\E_1+\E_2)-(m_1+m_2)},
\]
then the maximal future development contains two disjoint black hole regions, each asymptotic to a Kerr spacetime with ADM parameters close to $(\E_i,\P_i,\J_i)$.
\item[\emph{(ii)}] \emph{Merger.}
If
\[
d_{12}\;\leq\;\ka^{-1}\,\frac{m_1m_2}{(\E_1+\E_2)-(m_1+m_2)},
\]
then the future event horizon is connected.
\end{enumerate}
\end{conj}
The remainder of the paper is organized as follows. Section \ref{secgluing} recalls the definition of ADM charges and the obstruction-free gluing results of \cite{MOT,MaoTao}. Section \ref{sec:Kerr} computes the localized ADM charges of the Kerr metric $\g_{m,a}$ in Kerr-Schild coordinates. Section \ref{sec:Poincare} shows that the $\ISO^+(1,3)$--orbit of $\g_{m,a}$ realizes all boosted Kerr data with prescribed ADM charges by identifying two Casimirs. Section \ref{sec:Cauchy} proves Theorem \ref{maintheorem} using the above ingredients and the well-prepared short-pulse slice construction from \cite{ShenWan2}.
\paragraph{Acknowledgments.} The authors thank Elena Giorgi, Sergiu Klainerman and J\'er\'emie Szeftel for their interest in this work. J.W. is supported by ERC-2023 AdG 101141855 BlaHSt.
\section{Initial data gluing}\label{secgluing}
Let $(\Si,g,k)$ be an initial data set that solves \eqref{Econstraint}. Introduce the new variables
\begin{align}
\begin{split}\label{dfhdfpi}
h_{ij}:= g_{ij}-e_{ij}-\de_{ij}\,\tr_e(g-e),\qquad\quad\pi_{ij}:= k_{ij}-\de_{ij}\,\tr_e k .
\end{split}
\end{align}
All traces, index increases, and contractions in the following are taken with respect to the Euclidean metric $e$. The inverse relations are
\begin{align}
\begin{split}\label{inversegk}
g_{ij}= \de_{ij}+h_{ij}-\frac12\de_{ij}\tr_e h,
\qquad\quad
k_{ij}= \pi_{ij}-\frac12\de_{ij}\tr_e\pi.
\end{split}
\end{align}
In these variables, the Einstein vacuum constraints \eqref{Econstraint} can be written schematically as
\begin{align}
\begin{split}\label{schematicconstraint}
P(h,\pi)=\Phi(h,\pi)
\end{split}
\end{align}
where $P$ is the leading linear part 
\[
P(h,\pi):=(\pr_i\pr_j h^{ij},\pr_i\pi^{ij}),
\qquad
\Phi(h,\pi):=(M(h,\pi),N^j(h,\pi)),
\]
and $\Phi$ collects quadratic in $(h,\pi,\pr h,\pr\pi)$:
\[
M(h,\pi)=h\c\pr^2h+\pr h\c\pr h+\pi\c\pi,
\qquad
N^j(h,\pi)=(h\c\pr\pi)^j+(\pr h\c\pi)^j .
\]
\subsection{Definition of charges relative to a domain}
Let $(g,k)$ be an asymptotically flat data on $\RRR^3$, written in canonical coordinates $x^i$ with a Euclidean background $e$. For any closed surface $S\subset\RRR^3$, we define localized ADM charges (fluxes) by
\begin{align}
\begin{split}\label{EPCJdf}
    \E[(g,k);S]
    &:=\frac12\int_{S}(\pr_ig_{ij}-\pr_jg_{ii})\nu^j\,dS,\\
    \P_i[(g,k);S]
    &:=\int_{S}(k_{ij}-\de_{ij}\tr_{e}k)\nu^j\,dS,\\
    \C_l[(g,k);S]
    &:=\frac12\int_{S}\Big(x_l(\pr_ig_{ij}-\pr_jg_{ii})
      -\de_{il}(g-e)_{ij}
      +\de_{jl}(g-e)_{ii}\Big)\nu^j\,dS,\\
    \J_l[(g,k);S]
    &:=\int_{S}(k_{ij}-\de_{ij}\tr_ek)Y_l^i\nu^j\,dS,
\end{split}
\end{align}
where $\nu$ denotes the outward Euclidean unit normal to $S$ and $Y_l^i:={\in_{lj}}^i x^j$. We collect these into the charge vector
\begin{align}\label{dfQ}
    \Q[(g,k);S]:=(\E,\P_1,\P_2,\P_3,\C_1,\C_2,\C_3,\J_1,\J_2,\J_3)[(g,k);S].
\end{align}
For $S=\pr B_r$, the ADM charges are defined by
\[
\Q_{ADM}[(g,k)]
:=\lim_{r\to\infty}\Q[(g,k);\pr B_r],
\]
whenever the limit exists.

Since we will work on annular regions, we also introduce averaged charges. Fix $\eta\in C_c^\infty(0,\infty)$ to satisfy
\begin{equation*}
\supp\eta\subset(1,2),
\qquad
\int_1^2\eta(r)\,dr=1,
\end{equation*}
and define, for $r>0$,
\begin{align*}
\eta_r(r'):=r^{-1}\eta(r^{-1}r').
\end{align*}
For $\Q=(\E,\P,\C,\J)$, we set for  $A_r:=B_{2r}\setminus\ov{B_r}$ that
\begin{align}
\begin{split}\label{eq:charge-avg}
\Q[(g,k);A_r]
:=\int_r^{2r}\eta_r(r')\Q[(g,k);\pr B_{r'}]\,dr'.
\end{split}
\end{align}
\subsection{Gluing theorems}
In this subsection, we record a rescaled annular gluing theorem and a conic gluing result for the vacuum constraint equations, adapted from Mao-Oh-Tao \cite{MOT} and Mao-Tao \cite{MaoTao} in the form needed here.

We first state a rescaled annular gluing theorem.
\begin{Aglu}[c.f. Theorem 1.7 of \cite{MOT}]\label{MOT1.7}
Given $s>\frac{3}{2}$, $\Ga>1$ and $r>0$, there exist constants
$\ep_{o}=\ep_{o}(s,\Ga)>0$, $\mu_o=\mu_o(s,\Ga)>0$ and
$C_o=C_o(s,\Ga)>0$ such that the following holds.
Let $(g_{in},k_{in})\in H^s\times H^{s-1}(A_r)$ and
$(g_{out},k_{out})\in H^s\times H^{s-1}(A_{32r})$ be solutions of
\eqref{Econstraint}.
Define $\De\Q=(\De\E,\De\P,\De\C,\De\J)\in\RRR^{10}$ by
\begin{equation}\label{eq:Delta-Qr}
\De\Q
=\Q[(g_{out},k_{out});A_{32r}]
-\Q[(g_{in},k_{in});A_r].
\end{equation}
Assume
\begin{align}
\begin{split}\label{eq:obs-free-unit:EPr}
\De\E &>|\De\P|,
\qquad
\frac{\De\E}{\sqrt{(\De\E)^2-|\De\P|^2}}<\Ga,\\
r^{-1}\De\E &<\ep_o^2,
\qquad
r^{-1}(|\De\C|+|\De\J|)<\mu_o\De\E,
\end{split}
\end{align}
and
\begin{equation}\label{eq:obs-free-unit:datar}
r^{-2}\|g_{in}-e\|_{H^s(A_r)}^2
+\|k_{in}\|_{H^{s-1}(A_r)}^2
+r^{-2}\|g_{out}-e\|_{H^s(A_{32r})}^2
+\|k_{out}\|_{H^{s-1}(A_{32r})}^2
<\mu_o\De\E .
\end{equation}
Then there exists $(g,k)\in H^s\times H^{s-1}(B_{64r}\setminus\ov{B_r})$
solving \eqref{Econstraint} such that
\begin{equation}\label{eq:obs-free-gluingr}
(g,k)=(g_{in},k_{in})\ \text{on }A_r,
\qquad
(g,k)=(g_{out},k_{out})\ \text{on }A_{32r},
\end{equation}
and
\begin{equation}\label{eq:obs-free-concr}
r^{-2}\|g-e\|_{H^s(B_{64r}\setminus\ov{B_r})}^2
+\|k\|_{H^{s-1}(B_{64r}\setminus\ov{B_r})}^2
<C_o\De\E .
\end{equation}
\end{Aglu}
\begin{proof}
Applying \cite[Theorem 1.7]{MOT} to the rescaled data $(g^{(r)}(x),k^{(r)}(x))=(g(rx),rk(rx))$, this concludes the proof.
\end{proof}
For the conic gluing result, we need the following right inverse operator.  We now define the $b$--Sobolev space.
\begin{df}\label{dfbsobolev}
For $s\in\NNN$, the $b$--Sobolev space $H^s_b(\RRR^3)$ is defined by the norm
\[
\|u\|^2_{H_b^s(\RRR^3)}:=\sum_{k\leq s}\|\xja^k\nab^k u\|_{L^2(\RRR^3)}^2.
\]
We extend the definition to $s\in\RRR$ by duality and interpolation. For $\ell\in\RRR$, we set $H_b^{s,\ell}:=\xja^{-\ell}H_b^s$. For our purpose, it's convenient to set $\XX_b^{s,\de} :=H_b^{s,\de}\times H_b^{s-1,\de+1}$.
\end{df}
\begin{prop}[Proposition 9 in \cite{MaoTao}]\label{MaoTao9}
Let
\begin{equation}\label{dfOmint}
\Om_{int} :=\ov{(C_{\om,\th}\cup B_1)\setminus(C_{\om,\th_0}\cup B_{\frac12})}.
\end{equation}
There exists a solution operator
\[
S_{int}:H_b^{s-2,\de+2}(\Om_{int}) \to \XX_b^{s,\de}(\Om_{int}) :=H_b^{s,\de}(\Om_{int})\times H_b^{s-1,\de+1}(\Om_{int}),
\]
for $s\in\RR$ and $\de<-\frac12$, such that for all
$f\in C_c^\infty(\Om_{int})$,
\[
\supp(S_{int}f)\subseteq\Om_{int}, \qquad PS_{int}f=f .
\]
\end{prop}
Finally, we state and prove a conic gluing theorem adapted to our setting, which is a slight modification of \cite[Theorem 2]{MaoTao} by Mao-Tao.
\begin{Cglu}\label{Conicgluingthm}
Let $0<\th_0<\th<\frac{\pi}{2}$, $\de<-\frac12$, and $\om\in\SSS^2$.
Suppose $(g_0,k_0)$ solves \eqref{Econstraint} in $C_{\om,\th}\cup B_1$ and
satisfies
\begin{equation}\label{g0k0est}
\|(g_0-e,k_0)\|_{\XX_b^{s,\de}(\Om_{int})}\le\ep,
\end{equation}
for $\ep>0$ sufficiently small.
Then there exists a solution $(g,k)$ of \eqref{Econstraint} on $\RRR^3$ such that
\[
(g,k)=
\begin{cases}
(g_0,k_0) & \mbox{ in }\; C_{\om,\th_0}\cup B_{\frac12},\\
(e,0) & \mbox{ in }\;\, \RRR^3\setminus(C_{\om,\th}\cup B_1),
\end{cases}
\]
and
\begin{equation}\label{g-ekest}
\|(g-e,k)\|_{\XX_b^{s,\de}(\Om_{int})}\les\ep .
\end{equation}
\end{Cglu}
\begin{proof}
    Let $\chi$ be a cut-off function that
    \begin{align*}
        \chi(x)=\left\{\;\begin{aligned}1 \quad\; &\mbox{ in }\; C_{\om,\th_0}\cup B_\frac{1}{2},\\
        0 \quad\; &\mbox{ in }\;\,\RRR^3\setminus(C_{\om,\th}\cup B_1).\end{aligned}\right.
    \end{align*}
    Let $(h_0,\pi_0)$ be associated to $(g_0,k_0)$ by \eqref{dfhdfpi}. We aim to find $(\hti,\pit)\in\XX^{s,\de}_b(\Om_{int})$ so that the following holds:
    \begin{align}\label{fixedpoint}
        P(\chi h_0+\hti,\chi\pi_0+\pit)=\Phi(\chi h_0+\hti,\chi\pi_0+\pit).
    \end{align}
    Let $C_0>0$ be a fixed constant, we define the following space
    \begin{align*}
        \XX:=\left\{(\hti,\pit)\in\XX_b^{s,\de}(\Om_{int})\Big/\,\big\|(\hti,\pit)\big\|_{\XX_b^{s,\de}(\Om_{int})}\leq C_0\ep\right\},
    \end{align*}
    and the following operator on $\XX$:
    \begin{align*}
        T(\hti,\pit):=S_{int}\left(\Phi(\chi h_0+\hti,\chi\pi_0+\pit)-P(\chi h_0,\chi\pi_0)\right).
    \end{align*}
    Thus, \eqref{fixedpoint} reduces to the following fixed point problem:
    \begin{align}\label{Tfixedpoint}
        (\hti,\pit)=T(\hti,\pit).
    \end{align}
    For any $(\hti,\pit)\in\XX$, we have from \eqref{g0k0est}
    \begin{align*}
        \|P(\chi h_0,\chi\pi_0)\|_{H_b^{s-2,\de+2}(\Om_{int})}\leq C\ep,\qquad\|\Phi(\chi h_0+\hti,\chi\pi_0+\pit)\|_{H_b^{s-2,\de+2}(\Om_{int})}\les\ep^2,
    \end{align*}
    where $C>0$ is independent of $C_0$. Applying Proposition \ref{MaoTao9}, we infer
    \begin{align*}
    \left\|S_{int}\left(\Phi(\chi h_0+\hti,\chi\pi_0+\pit)-P(\chi h_0,\chi\pi_0)\right)\right\|_{\XX_b^{s,\de}(\Om_{int})}\leq C\ep.
    \end{align*}
    Thus, we have for $C_0$ large enough that $T(\XX)\subseteq\XX$. Next, we have from Proposition \ref{MaoTao9}
\begin{align*}
\left\|T(h_1,\pi_1)-T(h_2,\pi_2)\right\|_{\XX_b^{s,\de}(\Om_{int})}
&=\left\|S_{int}\left(\Phi(\chi h_0+h_1,\chi\pi_0+\pi_1)
-\Phi(\chi h_0+h_2,\chi\pi_0+\pi_2)\right)\right\|_{\XX_b^{s,\de}(\Om_{int})}\\
&\les
\left\|\Phi(\chi h_0+h_1,\chi\pi_0+\pi_1)
-\Phi(\chi h_0+h_2,\chi\pi_0+\pi_2)\right\|_{H_b^{s-2,\de+2}(\Om_{int})}\\
&\les
\ep\,\|(h_1-h_2,\pi_1-\pi_2)\|_{\XX_b^{s,\de}(\Om_{int})}.
\end{align*}
    Hence, $T$ is a contraction map on $\XX$. By the Banach fixed point theorem, there exists a unique $(\hti_*,\pit_*)\in\XX$ such that \eqref{Tfixedpoint} holds. We define
    \begin{align*}
        (h,\pi)=\left\{\begin{aligned}
            (h_0,\pi_0),\qquad &\mbox{ in }\; C_{\om,\th_0}\cup B_\frac{1}{2},\\
            (\chi h_0+\hti_*,\chi\pi_0+\pit_*),\qquad&\mbox{ in }\; \Om_{int},\\
            (0,0),\qquad&\mbox{ in }\, \RRR^3\setminus(C_{\om,\th}\cup B_1).
        \end{aligned}\right.
    \end{align*}
    Let $(g,k)$ be defined by $(h,\pi)$ via \eqref{inversegk}. Since \eqref{g-ekest} follows directly from the construction, this concludes the proof.
\end{proof}
\section{Localized ADM charges for Kerr initial data}\label{sec:Kerr}
In Kerr-Schild coordinates $(t,x,y,z)$, the Kerr metric $\g=\g_{m,a}$ takes the following form:
\begin{align}\label{defgmunu-Kerr}
\g_{\mu\nu}
=\etabf_{\mu\nu}+2H\ellbf_\mu\ellbf_\nu,
\end{align}
where $\etabf$ is the Minkowski metric. The general Kerr-Schild identities and the associated ADM decomposition used below are recorded in Appendix \ref{app:KS}. More precisely, the Kerr coefficients are given by\footnote{See, for instance, \cite[Section 2.1]{CCI} or \cite[Appendix A]{MOT}.}
\begin{align}\label{KerrSchildcomponents}
H=\frac{m\rt^3}{\rt^4+a^2z^2},
\qquad\quad\ellbf:=\left(1,\frac{\rt x+ay}{\rt^2+a^2},\frac{\rt y-ax}{\rt^2+a^2},\frac{z}{\rt}\right),
\end{align}
where $\rt=\rt(x,y,z)>0$ is defined implicitly by
\begin{align}\label{dfrt}
\frac{x^2+y^2}{\rt^2+a^2}+\frac{z^2}{\rt^2}=1.
\end{align}
\begin{df}\label{KerrO}
For a tensor field $X$, we write $X=\O^p_q$ if
\[
|\pr^l X|
\les
\frac{(m+|a|)^p}{r^{q+l}},
\qquad
\forall\,l\in\NNN,
\]
where $r:=\sqrt{x^2+y^2+z^2}$ is the Euclidean radius.
\end{df}
The purpose of this section is to compute the localized ADM charges of the initial data induced by $\g_{m,a}$ on $\Si_0:=\{t=0\}$. These localized fluxes capture the leading special relativistic charges of Kerr while retaining precise control of lower-order error terms, which will be essential for describing their behavior under asymptotic Poincar\'e transformations in Section \ref{sec:Poincare}.
\begin{prop}\label{KerrADM}
Let $(\Si_0,g,k)=(\RRR^3,g_{m,a},k_{m,a})$ be the initial data induced by $\g_{m,a}$ on $\Si_0$. Then, the ADM fluxes on the coordinate spheres $\pr B_r\subset\Si_0$ satisfy
\begin{align*}
\E[(g,k);\pr B_r]
&=8\pi m+\O_2^3,
\qquad
\P[(g,k);\pr B_r]=\O_2^3,
\\
\C[(g,k);\pr B_r]
&=\O_1^3,
\qquad\qquad\quad\,\,
\J[(g,k);\pr B_r]=8\pi am\,\ev_z+\O_1^3.
\end{align*}
In particular, the leading terms coincide with the special relativistic energy and angular momentum of a spinning particle of mass $m$ and spin $am\ev_z$, while the remaining quantities decay at the expected rates.
\end{prop}
\subsection{Localized ADM energy and center of mass}
We first compute the even-parity charges ($\E$ and $\C$), which depend only on the asymptotic behavior of the metric. Before computing the localized $\E[(g,k);\pr B_r]$ and $\C[(g,k);\pr B_r]$, we deduce the following basic identities, which will be used throughout this section.
\begin{lem}\label{preidentities}
Let $r:=\sqrt{x^2+y^2+z^2}$ be the Euclidean radius and let $\a:=(1+2H)^{-\frac{1}{2}}$ be the lapse function. Then, we have the following identities:
\begin{align*}
H&=\frac{m}{r}(1+\O_2^2),\qquad \ellbf_i\pr_iH=-\frac{m}{r^2}+\O_4^3,\qquad\quad \frac{x_j}{r}\pr_jH=-\frac{m}{r^2}+\O_4^3,\\
\frac{x_j\ellbf_j}{r}&=1+\O_2^2,\qquad\qquad\;\;\,\pr_i\ellbf_i=\frac{2}{r}+\O_3^2,\qquad\quad\;\,\ellbf_i\pr_i\ellbf_j \frac{x_j}{r}=\O_3^2,\\
\a&=1-\frac{m}{r}+\O_2^2.
\end{align*}
\end{lem}
\begin{proof}
    See Appendix \ref{proof-preidentities}.
\end{proof}
\begin{prop}\label{computeEC}
We have the following identities for $(\Si_0,g,k)=(\RRR^3,g_{m,a},k_{m,a})$:
\begin{align*}
\E[(g,k);\pr B_r]=8\pi m+\O_2^3,\qquad\quad\C[(g,k);\pr B_r]=\O_1^3.
\end{align*}
\end{prop}
\begin{proof}
We first expand the integrand for the ADM energy $\ek:=(\pr_ig_{ij}-\pr_jg_{ii})\nu^j$
\begin{align*}
\ek&=\big(\pr_i(2H\ellbf_i\ellbf_j)-\pr_j(2H\ellbf_i\ellbf_i)\big)\frac{x_j}{r}\\
&=2(\ellbf_i\pr_iH)\frac{x_j\ellbf_j}{r}+2H(\pr_i\ellbf_i)\frac{x_j\ellbf_j}{r}+2H (\ellbf_i\pr_i\ellbf_j)\frac{x_j}{r}-2\pr_jH\frac{x_j}{r},
\end{align*}
where we used the fact that $\sum_{i=1}^3\ellbf_i^2=1$. Applying Lemma \ref{preidentities}, we obtain
\begin{align}\label{ekcompute}
\ek=\left[-\frac{2m}{r^2}+\O_4^3\right](1+\O_2^2)+\frac{2m}{r}\left[\frac{2}{r}+\O_3^2\right](1+\O_2^2)+\frac{2m}{r^2}+\O_4^3=\frac{4m}{r^2}+\O_4^3.
\end{align}
Integrating it on $\pr B_r$, we deduce
\begin{align*}
    \E[(g,k);\pr B_r]=\frac{1}{2}\int_{\pr B_r}\ek dS=8\pi m+\O_2^3.
\end{align*}
Next, we write the integrand for the ADM center of mass $\ck_l$
\begin{align*}
\ck_l&:=\left(x_l\pr_ig_{ij}-x_l\pr_jg_{ii}-\de_{il}(g-e)_{ij}+\de_{jl}(g-e)_{ii}\right)\nu^j=x_l\ek-2H\ellbf_l\ellbf_j\frac{x_j}{r}+2H\frac{x_l}{r}.
\end{align*}
Recall from Lemma \ref{preidentities} and \eqref{ekcompute}
\begin{align*}
H=\frac{m}{r}(1+\O_2^2),\qquad\quad\ellbf_j\frac{x_j}{r}=1+\O_2^2,\qquad\quad\ek=\frac{4m}{r^2}(1+\O_2^2).
\end{align*}
Therefore, we obtain
\begin{align*}
\ck_l=\frac{4mx_l}{r^2}(1+\O_2^2)-\frac{2m\ellbf_l}{r}(1+\O_2^2)+\frac{2mx_l}{r^2}(1+\O_2^2)=\left[\frac{6mx_l}{r^2}-\frac{2m\ellbf_l}{r}\right](1+\O_2^2).
\end{align*}
Taking $l=1,2,3$ and applying \eqref{ellbfl}, we infer
\begin{align*}
    \ck_1=\left[\frac{4mx}{r^2}-\frac{2amy}{r^3}\right](1+\O_2^2),\qquad\ck_2=\left[\frac{4my}{r^2}+\frac{2amx}{r^3}\right](1+\O_2^2),\qquad\ck_3=\frac{4mz}{r^2}(1+\O_2^2).
\end{align*}
Integrating $\ck_l$ on $\pr B_r$, we deduce
\begin{align*}
    \C_l[(g,k);\pr B_r]&=\frac{1}{2}\int_{\pr B_r}\ck_ldS=\O^3_1.
\end{align*}
This concludes the proof of Proposition \ref{computeEC}.
\end{proof}
\subsection{Localized ADM momentum and angular momentum}
Next, we compute the odd-parity charges ($\P$ and $\J$), which depend essentially on the second fundamental form. Before computing the localized fluxes $\P[(g,k);\pr B_r]$ and $\J[(g,k);\pr B_r]$, we record two auxiliary lemmas. Their proofs are deferred to Appendices \ref{proof-expressionk} and \ref{proof-moreidentities}.
\begin{lem}\label{expressionk}
We have the following expressions:
    \begin{align*}
\a^{-1}k_{ij}&=2H\ellbf_l\pr_l(H\ellbf_i\ellbf_j)+\pr_j(H\ellbf_i)+\pr_i(H\ellbf_j),\\
\a^{-1}\tr_ek&=2(1+H)\ellbf_l\pr_lH+2H\pr_l\ellbf_l.
    \end{align*}
\end{lem}
\begin{lem}\label{moreidentities}
We have the following identities for $j=1,2,3$:
    \begin{align*}
        \ellbf_i\pr_i\ellbf_j=\O_3^2,\qquad
        \frac{x_i}{r}\pr_i\ellbf_1=-\frac{ay}{r^3}+\O_3^2,\qquad \frac{x_i}{r}\pr_i\ellbf_2=\frac{ax}{r^3}+\O_3^2,\qquad \frac{x_i}{r}\pr_i\ellbf_3=\O_3^2.
    \end{align*}
\end{lem}
\begin{prop}\label{computePJ}
We have the following identities for $(\Si_0,g,k)=(\RRR^3,g_{m,a},k_{m,a})$:
    \begin{align*}
    \P[(g,k);\pr B_r]=\O_2^3,\qquad\quad\J[(g,k);\pr B_r]=8\pi am\e_z+\O_1^3.
    \end{align*}
\end{prop}
\begin{proof}
We have from Lemma \ref{expressionk}
\begin{align*}
\frac{1}{\alpha}k_{ij}\nu^j&=2H\ellbf_l\pr_l(H\ellbf_i\ellbf_j)\frac{x_j}{r} +\pr_j(H\ellbf_i) \frac{x_j}{r} +\pr_i(H\ellbf_j)\frac{x_j}{r}\\
&=\frac{2H}{r}\ellbf_l\pr_l(H\ellbf_jx_j\ellbf_i)-\frac{2H^2}{r}\ellbf_i+\frac{x_j\ellbf_i}{r}\pr_jH+H\left(\frac{x_j}{r}\pr_j\ellbf_i\right)+\frac{1}{r}\pr_i(H\ellbf_jx_j)-\frac{H\ellbf_i}{r}\\
&=\frac{2H}{r}\ellbf_l\pr_l(H\ellbf_jx_j\ellbf_i)+\left(-\frac{2H^2}{r}-\frac{H}{r}+\frac{x_j}{r}\pr_jH\right)\ellbf_i+H\left(\frac{x_j}{r}\pr_j\ellbf_i\right)+\frac{1}{r}\pr_i(H\ellbf_jx_j).
\end{align*}
Next, we compute using Lemmas \ref{preidentities} and \ref{moreidentities}
\begin{align*}
    \frac{2H}{r}\ellbf_l\pr_l(H\ellbf_jx_j\ellbf_i)&=\frac{2m}{r^2}(1+\O_2^2)\ellbf_l\pr_l\left( m \ellbf_i(1+\O_2^2)\right) =\frac{2m^2}{r^2}(1+\O_2^2)\ellbf_l\pr_l\ellbf_i+\O_5^4=\O_5^4,\\
    -\frac{2H^2}{r}-\frac{H}{r}+ \frac{x_j}{r}\pr_jH&=-\left(\frac{2m^2}{r^3}+\frac{2m}{r^2}\right)(1+\O_2^2)=-\frac{2m}{r^2}-\frac{2m^2}{r^3}+\O_4^3,\\
    \frac{1}{r}\pr_i(H\ellbf_jx_j)&=\frac{m}{r}\pr_i(1+\O_2^2)=\O_4^3.
\end{align*}
Thus, we infer
\begin{align*}
    \frac{1}{\a}k_{ij}\nu^j=\left(-\frac{2m}{r^2}-\frac{2m^2}{r^3}\right)\ellbf_i+H\left(\frac{x_j}{r}\pr_j\ellbf_i\right)+\O_4^3.
\end{align*}
We also have from Lemma \ref{expressionk}
\begin{align*}
    \frac{1}{\a}\tr_ek&=2(1+H)\ellbf_l\pr_lH+2H\pr_l\ellbf_l\\
    &=2\left(1+\frac{m}{r}+\O_3^3\right)\left(-\frac{m}{r^2}+\O_4^3\right)+\left(\frac{2m}{r}+\O_3^3\right)\left(\frac{2}{r}+\O_3^2\right)\\
    &=\frac{2m}{r^2}-\frac{2m^2}{r^3}+\O_4^3.
\end{align*}
Denoting the integrand of $\P_i$ as $\pk_i:=(k_{ij}-\de_{ij}\tr_ek)\nu^j$, we obtain
\begin{align*}
    \a^{-1}\pk_i&=\a^{-1}k_{ij}\nu^j-\a^{-1}\frac{x_i}{r}\tr_e k\\
    &=\left(-\frac{2m}{r^2}-\frac{2m^2}{r^3}\right)\ellbf_i+H\left(\frac{x_j}{r}\pr_j\ellbf_i\right)-\frac{x_i}{r}\left(\frac{2m}{r^2}-\frac{2m^2}{r^3}\right)+\O_4^3\\
    &=\left(-\frac{2m}{r^2}-\frac{2m^2}{r^3}\right)\ellbf_i-\frac{2mx_i}{r^3}+\frac{2m^2x_i}{r^4}+\frac{mx_j}{r^2}\pr_j\ellbf_i+\O_4^3.
\end{align*}
Applying \eqref{ellbfl} and Lemma \ref{moreidentities}, we deduce
\begin{align*}
    \a^{-1}\pk_1&=\left(-\frac{2m}{r^2}-\frac{2m^2}{r^3}\right)\left(\frac{x}{r}+\frac{ay}{r^2}\right)-\frac{2mx}{r^3}+\frac{2m^2x}{r^4}-\frac{may}{r^4}+\O_4^3,\\
    &=-\frac{4mx}{r^3}-\frac{3may}{r^4}+\O_4^3.
\end{align*}
Similarly, we have
\begin{align*}
    \a^{-1}\pk_2=-\frac{4my}{r^3}+\frac{3max}{r^4}+\O_4^3,\qquad
    \a^{-1}\pk_3=-\frac{4mz}{r^3}+\O_4^3.
\end{align*}
Combining with the fact that $\a=1-\frac{m}{r}+\O_2^2$, we infer
\begin{align*}
    \pk_1&=-\frac{4mx}{r^3} +\frac{4m^2x}{r^4}-\frac{3amy}{r^4}+\O_4^3,\qquad\qquad
    \pk_2=-\frac{4my}{r^3} +\frac{4m^2y}{r^4}+\frac{3amx}{r^4}+\O_4^3,\\
    \pk_3&=-\frac{4mz}{r^3} +\frac{4m^2z}{r^4}+\O_4^3.
\end{align*}
Integrating them on $\pr B_r$, we deduce for $i=1,2,3$
\begin{align*}
    \P_i[(g,k);\pr B_r]&=\int_{\pr B_r}\pk_i\, dS=\O_2^3.
\end{align*}
Next, denoting the integrand of $\J_l$ as $\jk_l:=\pk_iY_l^i$, we have
\begin{align*}
    \jk_1&=\pk_3x_2-\pk_2x_3=-\frac{3amxz}{r^4}+\O_3^3,\qquad\qquad
    \jk_2=\pk_1x_3-\pk_3x_1=-\frac{3amyz}{r^4}+\O_3^3,\\
    \jk_3&=\pk_2x_1-\pk_1x_2=\frac{3am(x^2+y^2)}{r^4}+\O_3^3.
\end{align*}
Integrating on $\pr B_r$, from
\[
\int_{\pr B_r}\frac{xz}{r^4}dS=0,\qquad \int_{\pr B_r}\frac{yz}{r^4}dS=0,\qquad \int_{\pr B_r}\frac{x^2+y^2}{r^4}dS=\frac{8\pi}{3},
\]
we obtain
\begin{align*}
    \J_1[(g,k);\pr B_r]&=\int_{\pr B_r}\jk_1dS=\O_1^3,\qquad\qquad
    \J_2[(g,k);\pr B_r]=\int_{\pr B_r}\jk_2dS=\O_1^3,\\
    \J_3[(g,k);\pr B_r]&=\int_{\pr B_r}\jk_3dS=8\pi am+\O_1^3.
\end{align*}
This concludes the proof of Proposition \ref{computePJ}.
\end{proof}
Combining Propositions \ref{computeEC} and \ref{computePJ}, this concludes the proof of Proposition \ref{KerrADM}. These explicit formulas in Proposition \ref{KerrADM} identify the leading ADM charges of Kerr in a form directly comparable with special relativistic energy--momentum and angular momentum, and will be used in the next section to describe their transformation under asymptotic Poincar\'e diffeomorphisms.
\section{Kerr spacetime under Poincar\'e transformation}\label{sec:Poincare}
The Einstein vacuum equations are diffeomorphism invariant, and in the asymptotically flat setting the asymptotic symmetry group can be viewed as the proper orthochronous Poincar\'e group. In this section we realize the induced Poincar\'e action on Kerr spacetimes by explicit coordinate transformations and identify the resulting orbit of Kerr initial data in Kerr-Schild coordinates.

At the level of localized fluxes, this action agrees with the exact special relativistic transformation laws up to lower--order error terms. We exploit this separation using finite--radius charge functionals, which capture both Poincar\'e covariance and the decay of the remaining terms.
\subsection{Charges for linearized Einstein vacuum equations}
Let $(\M,\etabf)$ be the Minkowski spacetime, and let $\dot{\g}=\g-\etabf$ be a smooth symmetric $2$--tensor on $\M$. We introduce
\[
\dot{\H}_{\a\b} :=\dot{\g}_{\a\b}-\frac12\etabf_{\a\b}\tr_{\etabf}\dot{\g}.
\]
The linearization of the Einstein tensor $\G_{\mu\nu}[\g]=\Ric_{\mu\nu}[\g]-\frac12\R[\g]\g_{\mu\nu}$ at $\etabf$ is given by
\begin{equation}\label{linearizedG}
\D_{\etabf}\G[\dot{\g}]_{\a\b}
=\frac12\Big( -\nab^\ga\nab_\ga\Hd_{\a\b} +\nab_\a\nab^\ga\Hd_{\ga\b} +\nab_\b\nab^\ga\Hd_{\ga\a} -\etabf_{\a\b}\nab^\ga\nab^\de\Hd_{\ga\de}
\Big).
\end{equation}
Let $\X$ be a Killing vector field of $(\M,\etabf)$. Associated with $\dot{\g}$ and $\X$, we define the $2$--form
\begin{align}\label{def:UUU}
{}^{(\X)}\UUU_{\a\b}[\dot{\g}]
:=\frac12\Big[
&(-\nab_\a\H_{\ga\b}
+\nab_\b\H_{\ga\a}
+\etabf_{\ga\a}\nab^\de\H_{\b\de}
-\etabf_{\ga\b}\nab^\de\H_{\a\de})\X^\ga
\nonumber\\
&\qquad
+\H_{\ga\a}\nab^\ga\X_\b
-\H_{\ga\b}\nab^\ga\X_\a
\Big].
\end{align}
A direct computation shows that
\[
\nab^\a\big({}^{(\X)}\UUU_{\a\b}[\dot{\g}]\big)
=\D_{\etabf}\G[\dot{\g}]_{\a\b}\X^\a.
\]
Consequently, for any domain $\Om\subset\M$ with boundary $\pr\Om$,
\begin{equation}\label{eq:stokes-linearized}
\int_{\pr\Om}\star{}^{(\X)}\UUU[\dot{\g}]
=\int_\Om d\big(\star{}^{(\X)}\UUU[\dot{\g}]\big)
=-\int_\Om\star\D_{\etabf}\G[\dot{\g}](\X,\cdot),
\end{equation}
where $\star$ denotes the Hodge operator of $\etabf$. Let $\Si=\{t=0\}$ be a spacelike hypersurface in canonical coordinates $(t,x^i)$, and let $S\subset\Si$ be a closed $2$--surface. For any Killing vector field $\X$ of $(\M,\etabf)$, we define the
\emph{linearized charge functional}
\begin{equation}\label{eq:charge-functional}
\QQQ[\dot{\g};\X;S]
:=\int_S\star{}^{(\X)}\UUU[\dot{\g}].
\end{equation}
In particular, given a spacetime $(\M,\g)$ with induced data $(g,k)$ on $\Si$, we define the linearized energy--momentum and angular momentum charges by
\begin{align}\label{dfPPPQQQ}
\PPP_\mu[\dot{\g};S]
:=\QQQ[\dot{\g};\pr_{x^\mu};S],
\qquad
\MMM_{\mu\nu}[\dot{\g};S]
:=\QQQ[\dot{\g};x_\mu\pr_{x^\nu}-x_\nu\pr_{x^\mu};S].
\end{align}
These charges coincide with the corresponding ADM fluxes:
\begin{equation}\label{PPPJJJEPCJ}
(\PPP_0,\PPP_i,\MMM_{0i},\MMM_{ij})[\dot{\g};S]
=\big(\E,\P_i,\C_i,{\in^l}_{ij}\J_l\big)[(g,k);S].
\end{equation}
The material above is standard; see, for example, Section 2.5 of \cite{MOT}. Finally, we state the behavior of these charges under Poincar\'e transformations. Let $\g=\g_{m,a}$ be the Kerr metric and consider the transformation
\begin{equation}\label{Lorentztransform}
{x'}^\mu=\La^\mu{}_\nu x^\nu+\xi^\mu .
\end{equation}
For $r>0$, denote
\[
S_r:=\{t=0\}\cap\{|x|=r\},
\qquad
S'_r:=\{t'=0\}\cap\{|x'|=r\}.
\]
Then the induced charges satisfy
\begin{align}\label{eq:Poincare-domain}
\begin{split}
\PPP_{\mu}[\dot{\g};S_r']
&=\La_{\mu}{}^{\nu}\PPP_{\nu}[\dot{\g};S_r]+O(r^{-1}),
\\
\MMM_{\mu\nu}[\dot{\g};S_r']
&=\La_\mu{}^\a\La_\nu{}^\b\MMM_{\a\b}[\dot{\g};S_r]
+(\xi_\mu\La_\nu{}^\a-\xi_\nu\La_\mu{}^\a)\PPP_\a[\dot{\g};S_r]
+O(r^{-1}),
\end{split}
\end{align}
where $\La_\mu{}^\nu=(\La^{-1})^\nu{}_\mu$. The Poincar\'e transformation law \eqref{eq:Poincare-domain} follows by taking $\a=1$ and $n=3$ in \cite[Proposition E.1]{CD} by Chru\'sciel-Delay.
\subsection{Kerr spacetime under Poincar\'e transformation}
Motivated by \eqref{PPPJJJEPCJ}, we encode the leading--order ADM charges $(\E,\P,\C,\J)$ as a pair $(\PPP,\MMM)$ transforming covariantly under asymptotic Poincar\'e diffeomorphisms, thereby separating the exact special relativistic transformation from lower--order corrections. This packaging is convenient for describing the Poincar\'e action and its orbit on Kerr initial data.
\begin{df}\label{df:physical-states}
A \emph{state} is a pair
\[
(\PPP,\MMM)\in\RRR^{1+3}\times\bigwedge^2\RRR^{1+3},
\]
encoded by the charges $(\E,\P,\C,\J)\in\RRR_+\times\RRR^3\times\RRR^3\times\RRR^3$ as follows:
\begin{align*}
\PPP_0=\E,\qquad \PPP_i=\P_i,\qquad
\MMM_{0i}=\C_i,\qquad
\MMM_{ij}=\in^{ijk}\J_k.
\end{align*}
The \emph{Pauli-Lubanski vector} associated with $(\PPP,\MMM)$ is defined by
\begin{equation}\label{eq:def-PL}
\WWW^\mu:=\frac12\in^{\mu\nu\rho\sigma}\PPP_\nu\MMM_{\rho\sigma},
\qquad \in^{0123}=1.
\end{equation}
\end{df}
\begin{lem}\label{lem:PL-components}
For any state $(\PPP,\MMM)$ determined by $(\E,\P,\C,\J)$, the Pauli--Lubanski
vector $\WWW^\mu=(\WWW^0,\W)$ satisfies
\[
\WWW^0 = \P_i\,\J_i,
\qquad
\W_i = -\E\,\J_i + \ep_{ijk}\P_j\C_k,
\qquad
\WWW_\mu\PPP^\mu = 0.
\]
\end{lem}
\begin{proof}
By \eqref{eq:def-PL}, we have
\[
\WWW^0= \frac12\in^{0\nu\rho\si}\PPP_\nu\MMM_{\rho\si}=\frac12\in^{0ijk}\PPP_i\MMM_{jk}=\frac12\in^{ijk}\P_i\in_{jk\ell}\J_\ell=\de_{\ell}^i\P_i\J_\ell=\P_i\J_i,
\]
where we used $\in^{ijk}\in_{jk\ell}=2\de^i_\ell$. Next, we have for $i\in\{1,2,3\}$,
\begin{align*}
\WWW^i&=\frac{1}{2}\in^{i\nu\rho\si}\PPP_\nu\MMM_{\rho\si}\\
&=\frac12\left(\in^{i0jk}\PPP_0\MMM_{jk}+\in^{ij0k}\PPP_j\MMM_{0k}+\in^{ijk0}\PPP_j\MMM_{k0}\right)\\
&=-\frac{\E}{2}\in^{ijk}\in_{jk\ell}\J_\ell+\frac{1}{2}\in^{ijk}\P_j\C_k+\frac{1}{2}\in^{ijk}\P_j\C_k\\
&=-\E\,\J^i+\in^{ijk}\P_j\C_k.
\end{align*}
Finally, we compute
\[
\PPP_\mu\WWW^\mu=\E\,\WWW^0+\P_i\W_i=\E(\P_i\,\J_i)+\P_i(-\E\,\J_i+\in_{ijk}\P_j\C_k)=0.
\]
This completes the proof of Lemma \ref{lem:PL-components}.
\end{proof}
\begin{df}\label{dfn:Poincare}
The \emph{proper orthochronous Lorentz group} is given by
\[
\SO^+(1,3)
:=\{\Lambda\in\GL(4,\RRR)\mid \Lambda^\top\etabf\Lambda=\etabf,\;\Lambda^0{}_0>0\}.
\] 
The \emph{proper orthochronous Poincar\'e group} is the semidirect product
\[
\ISO^+(1,3):= \RRR^{1+3}\rtimes\SO^+(1,3),
\]
Elements of $\mathrm{ISO}^+(1,3)$ are pairs $(\xi,\La)$ with
\[
\xi\in\RRR^{1+3},\qquad\La\in\mathrm{SO}^+(1,3),
\]
and the group law is given by
\[
(\xi,\La)\cdot(\eta,\Ga)= (\xi+\La\eta,\,\La\Ga).
\]
The inverse of $(\xi,\La)$ is given by
\[
(\xi,\La)^{-1}=(-\La^{-1}\xi,\,\La^{-1}).
\]
The action of $(\xi,\La)\in\ISO^+(1,3)$ on spacetime $\RRR^{1+3}$ is defined by
\[x\mapsto\La x+\xi.\]
We introduce the following representation of $\ISO^+(1,3)$ on $\RRR^{1+3}\times\bigwedge^2\RRR^{1+3}$:\footnote{Note that the representation defined in \eqref{actionPPPMMM} coincides with the \emph{coadjoint representation} of $\ISO^+(1,3)$.}
\begin{equation}\label{actionPPPMMM}
(\xi,\La)\c(\PPP,\MMM):=\left(\La_\mu{}^\nu\PPP_\nu,\;\La_\mu{}^\a\La_\nu{}^\b\MMM_{\a\b}+(\xi_\mu\La_\nu{}^\a-\xi_\nu\La_\mu{}^\a)\PPP_\a\right).
\end{equation}
\end{df}
We now prove the following lemma, which identifies two invariant quantities.
\begin{lem}\label{lem:invariants}
Let $(\xi,\La)\in\ISO^+(1,3)$. Defining $(\PPP',\MMM'):=(\xi,\La)\c (\PPP,\MMM)$ and letting $\WWW'=\WWW'(\PPP',\MMM')$ be the Pauli-Lubanski vector of $(\PPP',\MMM')$, we have the invariance of the two Casimirs
\[
\PPP'_{\mu}\PPP'^{\mu}=\PPP_{\mu}\PPP^{\mu},\qquad\quad\WWW'_{\mu}{\WWW'}^{\mu}=\WWW_{\mu}\WWW^{\mu}.
\]
\end{lem}
\begin{proof}
By \eqref{actionPPPMMM}, we have
\begin{equation}\label{PPPconservation}
\PPP'_\mu\PPP'^\mu=\etabf_{\mu\nu}\PPP'^\mu\PPP'^\nu= \etabf_{\mu\nu}\La^\mu{}_\rho\La^\nu{}_\si\PPP^\rho\PPP^\si=\etabf_{\rho\si}\PPP^\rho\PPP^\si=\PPP_\mu\PPP^\mu,
\end{equation}
where we used $\La^\top\etabf\La=\etabf$. Next, we have from Definition \ref{df:physical-states} and \eqref{actionPPPMMM}
\begin{align*}
\WWW'^\mu=\frac{1}{2}\in^{\mu\nu\rho\sigma}\PPP'_\nu\MMM'_{\rho\si}=\frac12\in^{\mu\nu\rho\si}\La_\nu{}^\a\PPP_\a\La_\rho{}^\b\La_\si{}^\ga\MMM_{\b\ga}+\frac{1}{2}\in^{\mu\nu\rho\si}(\La\PPP)_\nu\big(\xi_\rho(\La\PPP)_\si-\xi_\si(\La\PPP)_\rho\big).
\end{align*}
Notice that the following identities are valid:
\begin{align*}
\in^{\mu\nu\rho\si}(\La\PPP)_\nu(\xi_\si(\La\PPP)_\rho -\xi_\rho (\La\PPP)_\si)&=0,\\
\in^{\mu\nu\rho\si}\La_\mu{}^\de\La_\nu{}^\a\La_\rho{}^\b\La_\si{}^\ga=\det(\La)\in^{\de\a\b\ga}&=\in^{\de\a\b\ga}.
\end{align*}
Thus, we obtain
\begin{align*}
\La_\mu{}^\de\WWW'^\mu=\frac12\in^{\de\a\b\ga}\PPP_\a\MMM_{\b\ga}=\WWW^\de,
\end{align*}
which implies $\WWW'^\mu=\La^\mu{}_\nu \WWW^\nu$. Proceeding as in \eqref{PPPconservation}, we deduce that $\WWW'_{\mu}{\WWW'}^{\mu}=\WWW_{\mu}\WWW^{\mu}$. This concludes the proof of Lemma \ref{lem:invariants}.
\end{proof}
\begin{rk}
The representation of the Lie group $\ISO^+(1,3)$ defined in \eqref{actionPPPMMM} induces, by differentiation at the identity, a representation of the Lie algebra $\iso(1,3):=\RRR^{1+3}\rtimes\so(1,3)$:
\begin{align*}
(c,\Om)\c(\PPP,\MMM)
=\Big(\Om_\mu{}^\a\PPP_\a,\;
\Om_\mu{}^\a\MMM_{\a\nu}+\Om_\nu{}^\a\MMM_{\a\mu}
+c_\mu\PPP_\nu-c_\nu\PPP_\mu\Big).
\end{align*}
Let $U(\iso(1,3))$ denote the universal enveloping algebra of $\iso(1,3)$ and let $Z(U(\iso(1,3)))$ be its center. It is isomorphic to a polynomial algebra generated by two algebraically independent Casimir elements, which may be represented by
\[
\PPP_\mu\PPP^\mu
\qquad\text{and}\qquad
\WWW_\mu\WWW^\mu,
\]
where $\WWW^\mu$ is the Pauli--Lubanski vector defined in \eqref{eq:def-PL}. These Casimir elements correspond to $\ISO^+(1,3)$--invariant polynomial functions on $\RRR^{1+3}\times\bigwedge^2\RRR^{1+3}$ and are therefore constant on each Poincar\'e coadjoint orbit.

Lemma \ref{lem:invariants} verifies this invariance directly for the action \eqref{actionPPPMMM}, while Proposition \ref{prop:orbit-massive-spinning} shows that, in the massive spinning case, these two invariants completely characterize the $\ISO^+(1,3)$--orbit. For a physics-oriented discussion of this coadjoint-orbit interpretation, see \cite{AGGKM}.
\end{rk}
\begin{prop}\label{prop:orbit-massive-spinning}
Let $(\PPP,\MMM)$ be a state encoded by the charges $(\E,\P,\C,\J)$. In view of the invariance in Lemma \ref{lem:invariants}, we introduce the following constraint set:
\[
\RR_{\PPP,\MMM}=\left\{(\PPP',\MMM')\big/\;\PPP'_0>0,\quad\PPP'_\mu\PPP'^\mu=\PPP_\mu\PPP^\mu,\quad\WWW'_\mu\WWW'^\mu=\WWW_\mu\WWW^\mu\right\}.
\]
We also denote $\Orb_{\PPP,\MMM}$ the $\ISO^+(1,3)$--orbit of $(\PPP,\MMM)$ in $\RRR^{1+3}\times\bigwedge^2\RRR^{1+3}$. Then, we have
\begin{align*}
    \Orb_{\PPP,\MMM}=\RR_{\PPP,\MMM}.
\end{align*}
\end{prop}
\begin{proof}
By Lemma \ref{lem:invariants}, we have $\Orb_{\PPP,\MMM}\subseteq\RR_{\PPP,\MMM}$. We now prove the inverse inclusion, which proceeds by successive Lorentz boost, spatial translation, and spatial rotation. Let $(\PPP',\MMM')$ be any state encoded by the charges $(\E',\P',\C',\J')$ satisfying
\begin{align*}
\PPP'_0>0,\qquad\PPP'_\mu\PPP'^\mu=\PPP_\mu\PPP^\mu,\qquad\WWW'_\mu\WWW'^\mu=\WWW_\mu\WWW^\mu.
\end{align*}
{\bf Lorentz boost.} We first define
\[
v_i:=\frac{\P'_i}{\E'},\qquad\quad\ga:=\frac{1}{\sqrt{1-|\vec{v}|^2}}=\frac{\E'}{\sqrt{\E'^2-|\P'|^2}}.
\]
Let $B(\vec{v})$ be the corresponding Lorentz boost defined by the following matrix\footnote{Here, $\vec{v}=(v_1,v_2,v_3)^\top$ is a $3$--dimensional column vector and $|\vec{v}|:=\sqrt{v_1^2+v_2^2+v_3^2}$.}
\begin{align*}
B(\vec{v}):=\begin{pmatrix}\ga &-\ga\vec{v}^\top\\ -\ga\vec{v} & I_3+(\ga-1)\frac{\vec{v}\vec{v}^\top}{|\vec{v}|^2}
\end{pmatrix}.
\end{align*}
We then compute
\begin{align*}
B(\vec{v})\PPP'&=\begin{pmatrix}
\ga&-\ga\vec{v}^\top\\-\ga\vec{v} & I_3+(\ga-1)\frac{\vec{v}\vec{v}^\top}{|\vec{v}|^2}
\end{pmatrix}\begin{pmatrix}
    \E' \\ \P'
\end{pmatrix}=\begin{pmatrix}
\ga(\E'-v^i\P'_i)\\
-\ga\vec{v}\,\E'+\big(I_3+(\ga-1)\frac{\vec{v}\vec{v}^\top}{|\vec{v}|^2}\big)\P'
\end{pmatrix}\\
&=\begin{pmatrix}
\ga\frac{\E'^2-|\P'|^2}{\E'} \\ -\ga\P'+\P'+(\ga-1)\P'
\end{pmatrix}=\begin{pmatrix}
    \sqrt{\E'^2-|\P'|^2} \\ 0
\end{pmatrix}=\begin{pmatrix}
    \sqrt{-\PPP_\mu\PPP^\mu} \\ 0
\end{pmatrix}.
\end{align*}
Denoting
\[
(\PPP_B,\MMM_B):=(0,B(\vec{v}))\c(\PPP',\MMM'),
\]
we have from Lemma \ref{lem:invariants} that
\[
\PPP_B=(\sqrt{-\PPP_\mu\PPP^\mu},0,0,0),\qquad\quad (\WWW_B)_\mu(\WWW_B)^\mu=\WWW_\mu\WWW^\mu.
\]
{\bf Space translation.} Denoting $(\E_B,\P_B,\C_B,\J_B)$ the encoded charges of the state $(\PPP_B,\MMM_B)$, we introduce the following translation:
\begin{align*}
    \xi^\mu:=\left(0,\frac{\C_B}{\E_B}\right).
\end{align*}
Denoting $(\PPP_T,\MMM_T):=(\xi,0)\c(\PPP_B,\MMM_B)$, we have from \eqref{actionPPPMMM} that $\PPP_T=\PPP_B$ and
\begin{align*}
    (\MMM_T)_{\mu\nu}=(\MMM_B)_{\mu\nu}+(\xi_\mu\de_{\nu}{}^\a-\xi_\nu\de_{\mu}{}^\a)(\PPP_B)_\a=(\MMM_B)_{\mu\nu}+(\xi^\mu\de_\nu{}^0-\xi^\nu\de_\mu{}^0)\E_B.
\end{align*}
Hence, we infer
\begin{align*}
    (\MMM_T)_{0i}=(\C_B)_i-\xi_i\E_B=0.
\end{align*}
By Lemma \ref{lem:invariants}, $(\PPP_T,\MMM_T)$ satisfies:
\begin{align*}
    \PPP_T=(\sqrt{-\PPP_\mu\PPP^\mu},0,0,0),\qquad\quad (\WWW_T)_\mu(\WWW_T)^\mu=\WWW_\mu\WWW^\mu.
\end{align*}
Denoting $(\E_T,\P_T,\C_T,\J_T)$ the charges of the state $(\PPP_T,\MMM_T)$, we have
\begin{align*}
    \E_T=\sqrt{-\PPP_\mu\PPP^\mu},\qquad\quad\P_T=0,\qquad\quad\C_T=0.
\end{align*}
{\bf Space rotation.} Since $\{(\MMM_T)_{ij}\}_{1\leq i,j\leq 3}$ is a $3\times 3$ antisymmetric matrix. There exists $R\in SO(3)$ such that
\begin{align}\label{canonicalformMT}
    R\MMM_TR^\top=\begin{pmatrix}
        0 & |\J_T| & 0\\
        -|\J_T| & 0 & 0 \\
        0 & 0 & 0
    \end{pmatrix}.
\end{align}
We then introduce $\La_R$ the Lorentz transform defined by
\begin{align*}
    \La_R=\begin{pmatrix}
        1 & 0 \\ 0 & R
    \end{pmatrix}.
\end{align*}
Denoting $(\PPP_R,\MMM_R):=(0,\La_R)\c(\PPP_T,\MMM_T)$, we have $\PPP_R=\PPP_T$ and
\begin{align*}
(\MMM_R)_{\mu\nu}=(\La_R)_\mu{}^\a(\La_R)_\nu{}^\b(\MMM_T)_{\a\b}.
\end{align*}
Thus, we infer
\begin{align*}
    (\MMM_R)_{0i}&=(\La_R)_0{}^0(\La_R)_i{}^j(\MMM_T)_{0j}=R_i{}^j(\C_T)_j=0,\\
    (\MMM_R)_{ij}&=R_i{}^kR_j{}^l(\MMM_T)_{kl}=(R\MMM_TR^\top)_{ij}.
\end{align*}
Letting $(\E_R,\P_R,\C_R,\J_R)$ the encoded charges of $(\PPP_R,\MMM_R)$, we deduce from \eqref{canonicalformMT}
\begin{align*}
    (\J_R)_1=0,\qquad\quad (\J_R)_2=0,\qquad\quad (\J_R)_3=|\J_T|.
\end{align*}
Combining the above identities, we infer
\begin{align*}
    \E_R=\sqrt{-\PPP_\mu\PPP^\mu},\qquad \P_R=0,\qquad \C_R=0,\qquad \J_R=(0,0,|\J_R|).
\end{align*}
Moreover, we have from Lemma \ref{lem:invariants}
\begin{align*}
    \WWW_\mu\WWW^\mu=(\WWW_R)_\mu\WWW_R^\mu=|\E_R\,\J_R|^2=-\PPP_\mu\PPP^\mu|\J_R|^2.
\end{align*}
Thus, we obtain
\begin{align*}
    \E_R=\sqrt{-\PPP_\mu\PPP^\mu},\qquad \P_R=0,\qquad\C_R=0,\qquad \J_R=\left(0,0,\sqrt{-\frac{\WWW_\mu\WWW^\mu}{\PPP_\mu\PPP^\mu}}\right).
\end{align*}
\paragraph{Conclusion.}
Combining the above steps, for any $(\PPP',\MMM')\in\RR_{\PPP,\MMM}$, there exists $(\xi,\La):=(0,\La_R)\circ(\xi,0)\circ(0,B(\vec v))\in\ISO^+(1,3)$ such that
\begin{align*}
(\xi,\La)\circ\PPP'=(\sqrt{-\PPP_\mu\PPP^\mu},0,0,0),\qquad
(\xi,\La)\circ\MMM'=
\begin{pmatrix}
0 & 0 & 0 & 0\\
0 & 0 & \sqrt{-\frac{\WWW_\mu\WWW^\mu}{\PPP_\mu\PPP^\mu}} & 0\\
0 & -\sqrt{-\frac{\WWW_\mu\WWW^\mu}{\PPP_\mu\PPP^\mu}} & 0 & 0\\
0 & 0 & 0 & 0
\end{pmatrix}.
\end{align*}
Applying the same construction to $(\PPP,\MMM)$, there exists
$(\xi_0,\La_0)\in\ISO^+(1,3)$ such that
\[
(\xi_0,\La_0)\circ(\PPP,\MMM)=(\xi,\La)\circ(\PPP',\MMM').
\]
Hence
\[
(\PPP',\MMM')=(\xi,\La)^{-1}\circ(\xi_0,\La_0)\circ(\PPP,\MMM)\in\Orb_{\PPP,\MMM},
\]
and therefore $\RR_{\PPP,\MMM}\subseteq\Orb_{\PPP,\MMM}$. This completes the proof of Proposition \ref{prop:orbit-massive-spinning}.
\end{proof}
\begin{prop}\label{mainergodic}
Let $\E_*>0$ and $\P_*,\J_*\in\RRR^3$ be fixed constants satisfying $\E_*>|\P_*|$. Then, there exists a Kerr metric $\g_{m,a}$ in Kerr-Schild coordinates $(t,x,y,z)$ and a Poincar\'e transform:
\begin{align}\label{Phiiso}
    x'^\mu=\La^\mu{}_\nu x^\nu+\xi^\mu
\end{align}
such that the following holds:
\begin{itemize}
\item Let $(g',k')$ be the induced initial data obtained by restricting $\g_{m,a}$ on $\Si'_0:=\{t'=0\}$. Then, we have
    \begin{align}\label{g'k'behavior}
        g'-e=O(r'^{-1}),\qquad\quad k'=O(r'^{-2}),
    \end{align}
where $r':=\sqrt{x'^2+y'^2+z'^2}$ denotes the Euclidean distance from $(x',y',z')$ to $\vec{0}$ and the constants involved in $O$ depend only on $\E_*$, $\P_*$ and $\J_*$.
\item For any $r>0$, we have the following identities on $S_r':=\{(t',x',y',z'):t'=0,r'=r\}$:
\begin{align}
\begin{split}\label{S_r'charges}
    \E[(g',k');S'_r]&=\E_*+O(r^{-1}),\qquad \P[(g',k');S'_r]=\P_*+O(r^{-1}),\\
    \C[(g',k');S'_r]&=O(r^{-1}),\qquad\qquad\,\,\,\,\J[(g',k');S'_r]=\J_*+O(r^{-1}).
\end{split}
\end{align}
\item The ADM charges of $(g',k')$ are given by:
\begin{align}
\begin{split}\label{ADMg'k'}
\E^{ADM}[(g',k')]&=\E_*,\qquad\quad\; \P^{ADM}[(g',k')]=\P_*,\\
\C^{ADM}[(g',k')]&=0,\qquad\qquad\,\J^{ADM}[(g',k')]=\J_*.
\end{split}
\end{align}
\end{itemize}
\end{prop}
\begin{proof}
We first define
\begin{align*}
m=\frac{1}{8\pi}\sqrt{\E_*^2-|\P_*|^2},\qquad\quad a=\frac{\sqrt{\E_*^2|\J_*|^2-(\P_*\c\J_*)^2}}{\E_*^2-|\P_*|^2}.
\end{align*}
Then, we denote
\begin{align*}
    \E=8\pi m,\qquad \P=0,\qquad \C=0,\qquad \J=8\pi am\e_z.
\end{align*}
Let $(\PPP,\MMM)$ and $(\PPP_*,\MMM_*)$ be the states of charges $(\E,\P,\C,\J)$ and $(\E_*,\P_*,\C_*,\J_*)$, respectively. Then, we have from Lemma \ref{lem:PL-components}
\begin{align*}
(\PPP_*)_\mu(\PPP_*)^\mu&=-\E_*^2+|\P_*|^2=-(8\pi m)^2=-\E^2=\PPP_\mu\PPP^\mu,\\
(\WWW_*)_\mu(\WWW_*)^\mu&=-(\P_*\c\J_*)^2+\E_*^2|\J_*|^2=(8\pi m)^4a^2=\E^2|\J|^2=\WWW_\mu\WWW^\mu.
\end{align*}
By Proposition \ref{prop:orbit-massive-spinning}, there exists $(\xi,\La)\in\ISO^+(1,3)$ such that $(\xi,\La)\c(\PPP,\MMM)=(\PPP_*,\MMM_*)$. Then, \eqref{Phiiso} follows immediately from \eqref{Lorentztransform}. Moreover, we have trivially
\begin{align*}
    \|\g_{m,a}-\etabf\|_{C^s}\les\frac{m+|a|}{r'},\qquad \forall\; s\in\NNN,
\end{align*}
which implies \eqref{g'k'behavior}. We recall from Proposition \ref{KerrADM}
\begin{align*}
\E[(g,k);\pr B_r]&=8\pi m+\O_2^3,\qquad\quad \P[(g,k);\pr B_r]=\O_2^3,\\
\C[(g,k);\pr B_r]&=\O_1^3,\qquad\qquad\qquad\;\,\J[(g,k);\pr B_r]=8\pi am\ev_z+\O_1^3,
\end{align*}
where $(g,k)$ denotes the induced initial data by restricting $\g_{m,a}$ on $\{t=0\}$. Using the charge identities \eqref{PPPJJJEPCJ} together with the Poincar\'e transformation formulas \eqref{eq:Poincare-domain}, we obtain
\begin{align*}
    \E[(g',k');S'_r]&=\E_*+O(r^{-1}),\qquad \P[(g',k');S'_r]=\P_*+O(r^{-1}),\\
    \C[(g',k');S'_r]&=O(r^{-1}),\qquad\qquad\,\,\,\,\J[(g',k');S'_r]=\J_*+O(r^{-1}).
\end{align*}
Taking $r\to\infty$, we obtain \eqref{ADMg'k'}. This concludes the proof of Proposition \ref{mainergodic}.
\end{proof}
\section{Proof of Theorem \ref{maintheorem}}\label{sec:Cauchy}
We begin by recalling a preliminary result from \cite{ShenWan2}, which constructs
a well–controlled spacelike short–pulse slice serving as the local interior
model in the proof of Theorem \ref{maintheorem}.
\begin{lem}[Theorem 4.27 in \cite{ShenWan2}]\label{interiorsolution}
For any $s\in\mathbb{N}$, there exists a sufficiently small $\ep>0$. For any $0<\de\leq \ep^2$ and $R>0$, there exists a spacelike initial data $\Si(\de,\ep,R):=(\Si_{\de,\ep,R},g_{\de,\ep,R},k_{\de,\ep,R})$ solving \eqref{Econstraint}, endowed with a radial $r$--foliation for $r\in(0,2R)$, which satisfies the following properties: \begin{enumerate}
\item We have
\begin{align}
\begin{split}\label{diffge}
&(g_{\de,\ep,R},k_{\de,\ep,R})=(e,0)\qquad\qquad\qquad\qquad\mbox{ in }\;B_{(1-2\de)R},\\
&R^{-\frac{3}{2}}\|g_{\de,\ep,R}-e\|_{H^s(A_R)}+R^{-\frac{1}{2}}\|k_{\de,\ep,R}\|_{H^{s-1}(A_R)}\les\ep.
\end{split}
\end{align}
\item Trapped surfaces will form in $D^+(B_R)$.
\end{enumerate}
Moreover, $A_R$ is called the \emph{barrier annulus}, and $B_R\setminus\ov{B_{(1-2\de)R}}$ is called the \emph{short-pulse annulus}.
\end{lem}
Lemma \ref{interiorsolution} provides a spacelike realization of the short-pulse mechanism inside a compact region, with quantitative control suitable for gluing. A closely related perspective is developed in the recent work of Chen-Klainerman \cite{ChenKlainerman}, where trapped surface formation is achieved purely at the level of spacelike Cauchy data, by identifying the appropriate freely prescribable components in the elliptic-transport formulation of the vacuum constraint equations.

We now glue in such an interior short-pulse region from \cite{ShenWan2} into an asymptotically flat initial data set with prescribed ADM parameters.
\begin{prop}\label{1BHformation}
Let $s\in\mathbb{N}$ and let $(\E_*,\P_*,\J_*)\in\RRR_+\times\RRR^3\times\RRR^3$ satisfy $\E_*>|\P_*|$. Then, there exist a sufficiently large $R>0$ and a sufficiently small $0<\ep\leq R^{-3}$ such that the following holds. For any $0<\de\leq \ep^2$, there exists a spacelike initial data $(\Si,g,k)$ solving \eqref{Econstraint}, endowed with a radial $r$--foliation for $r>0$, which satisfies the following properties:
\begin{enumerate}
\item We have
    \begin{align}
    \begin{split}
        &(g,k)=(e,0)\qquad\qquad\qquad\qquad\qquad\qquad\qquad\, \mbox{ in }\;B_{(1-2\de)R},\\
        &R^{-1}\|g-e\|_{H^s(B_{64R}\setminus\ov{B_R})}+\|k\|_{H^{s-1}(B_{64R}\setminus\ov{B_R})}\les 1,\\
        &(g,k)=(g',k')\qquad\qquad\qquad\qquad\qquad\qquad\quad\;\;\mbox{ in }\;B_{32R}^c,
    \end{split}
    \end{align}
    with $(g',k')$ the initial data obtained in Proposition \ref{mainergodic} with parameters $(\E_*,\P_*,\J_*)$.
    \item Trapped surfaces will form in $D^+(B_R)$.
\end{enumerate}
\end{prop}
\begin{proof}
We denote
\begin{align*}
(g_{in},k_{in}):=(g_{\de,\ep,R},k_{\de,\ep,R}),\qquad\quad (g_{out},k_{out}):=(g',k').
\end{align*}
By Lemma \ref{interiorsolution}, we have
\begin{align*}
    \E[(g_{in},k_{in});A_R]+|\P[(g_{in},k_{in});A_R]|&\les\ep R\les R^{-1},\\
    |\C[(g_{in},k_{in});A_R]|+|\J[(g_{in},k_{in});A_R]|&\les\ep R^2\les R^{-1}\\
    R^{-2}\|g_{in}-e\|_{H^s(A_R)}^2+\|k_{in}\|_{H^{s-1}(A_R)}^2&\les\ep^2 R\les R^{-1}\ll\De\E.
\end{align*}
We also have from Proposition \ref{mainergodic}
\begin{align*}
    \E[(g_{out},k_{out});A_{32R}]&=\E_*+O(R^{-1}),\\ \P[(g_{out},k_{out});A_{32R}],\;\C[(g_{out},k_{out});A_{32R}]&=O(R^{-1}),\\
    \J[(g_{out},k_{out});A_{32R}]&=\J_*+O(R^{-1}).
\end{align*}
It follows that
\begin{align*}
\De\E=\E_*+O(R^{-1}),\qquad \De\P=\P_*+O(R^{-1}),\qquad \De\C=O(R^{-1}),\qquad \De\J=\J_*+O(R^{-1}).
\end{align*}
Thus, we obtain for $\Ga:=\frac{2\E_*}{\sqrt{\E_*^2-|\P_*|^2}}$ and $R\gg 1$:
\begin{align*}
    |\De\P|<|\De \E|,\qquad \frac{\De\E}{\sqrt{|\De\E|^2-|\De\P|^2}}<\Ga,\qquad R^{-2}(|\De\C|+|\De\J|)\ll R^{-1}\De\E\ll 1.
\end{align*}
We also have
\begin{align*}
R^{-2}\|g_{out}-e\|_{H^s(A_{32R})}^2+\|k_{out}\|_{H^{s-1}(A_{32R})}^2\les\int_{A_{32R}}R^{-4}dx\les R^{-1}\ll\De\E.
\end{align*}
Hence, all conditions in \refAglu are valid. Then, there exists $(g, k) \in H^{s} \times H^{s-1}(B_{64R} \setminus \overline{B_{R}})$ that solves \eqref{Econstraint} and
\begin{equation*}
(g,k) = (g_{in},k_{in})\quad \mbox{ on }\; A_R,\qquad\quad (g,k)=(g_{out},k_{out}) \quad \mbox{ on }\; A_{32R}.
\end{equation*}
Moreover, we have
\begin{equation*}
R^{-2}\|g-e\|_{H^{s}(B_{64R}\setminus\ov{B_{R}})}^{2}+\|k\|_{H^{s-1}(B_{64R}\setminus\ov{B_{R}})}^2\les\De\E\les 1.
\end{equation*}
This concludes the proof of Proposition \ref{1BHformation}.
\end{proof}
We are now ready to prove the main theorem.
\begin{proof}[Proof of Theorem \ref{maintheorem}]
The proof is divided into 2 steps.
\paragraph{Step 1. Construction of Cauchy data.} We define $y_I=C\om_I$ for all $I=1,2,\ldots,N$ with $C\gg 1$ as a fixed constant such that
    \begin{equation*}
        B_1(y_I)\cap B_1(y_J)=\emptyset,\qquad\quad\forall\; I\ne J.
    \end{equation*}
Then, by construction, we have that $\{C_{\om_I,\th_I}(y_I)\}_{I=1}^N$ is mutually disjoint for $I=1,2,\ldots,N$. Next, by Proposition \ref{mainergodic}, for any $\cb_I\in\RRR^3$, there exists initial data $(g_I,k_I)$ that solves \eqref{Econstraint} such that the following holds:
    \begin{align}\label{gIkIbehavior}
        g_I-e=O(r_I^{-1}),\qquad\quad k_I=O(r_I^{-2}),\qquad\quad r_I:=|\x-\cb_I|.
    \end{align}
We also have from Proposition \ref{1BHformation} that there exists an initial data set $(g_I',k_I')$ such that:
\begin{align}
        \begin{split}\label{gIkI'est}
        (g_I',k_I')&=(e,0)\qquad\mbox{ in }\;B_{(1-2\de_I)R_I},\\
        (g_I',k_I')&=(g_I,k_I)\quad\;\mbox{ in }\;B_{32R_I(\cb_I)}^c,
        \end{split}
    \end{align}
    and
    \begin{equation}\label{gIkI'Sobolev}
        R_I^{-1}\|g_I'-e\|_{H^s(B_{64R_I}(\cb_I)\setminus\ov{B_R(\cb_I)})}+\|k_I'\|_{H^{s-1}(B_{64R_I}(\cb_I)\setminus\ov{B_R(\cb_I)}}\les 1,
    \end{equation}
    with $0<\de_I\ll R_I^{-6}\ll 1$. We now fix $\cb_I$ by defining $\cb_I:=y_I+C_IR_I\om_I$ with $C_I\gg 1$ such that $B_{64R_I}(\cb_I)\subset C_{\om_I,\frac{1}{2}\th_I}(y_I)$. We then have from \eqref{defOmI} and \eqref{gIkIbehavior} that
    \begin{align*}
        g_I'-e=O(r_I^{-1}),\qquad\quad k_I'-e=O(r_I^{-2})\qquad \mbox{ in }\;\Om_I,
    \end{align*}
    where $\Om_I$ is defined in \eqref{defOmI}. Thus, we infer for any $\de<-\frac{1}{2}$
    \begin{align*}
    \|(g'_I-e,k'_I)\|_{H_b^{s,\de}\times H_b^{s-1,\de+1}(\Om_I)}^2\les  \int_{\Om_I} r_I^{-2+2\de}=\int_{\{r_I\geq 64R_I\}}^\infty r_I^{-2+2\de}\les R_I^{1+2\de}\ll 1.
    \end{align*}
    Applying \refCglu with $\ep=R_I^{\de+\frac{1}{2}}$, there exists $(g_I'',k_I'')$ that solves \eqref{Econstraint} such that
    \begin{align*}
    (g_I'',k_I'')=\left\{
    \begin{aligned}
    (g_I',k_I'),\qquad&\mbox{ in }\;C_{\om_I,\frac{1}{2}\th_I}(y_I)\cup B_\frac{1}{2}(y_I),\\
    (e,0),\qquad&\mbox{ in }\;\RRR^3\setminus(C_{\om_I,\th_I}(y_I)\cup B_1(y_I)),
    \end{aligned}\right.
\end{align*}
and the following estimate holds:
\begin{align}\label{gIkI''est}
\|(g''_I-e,k''_I)\|_{H_b^{s,\de}\times H_b^{s-1,\de+1}(\Om_I)}\les R_I^{\de+\frac{1}{2}}\les 1.
\end{align}
We then define the desired Cauchy data $(g,k)$ as follows:
\begin{align*}
    (g,k)=\left\{
    \begin{aligned}
    (g_I'',k_I'')\qquad&\mbox{ in }\; \ov{C_{\om_I,\th_I}(y_I)\cup B_1(y_I)},\qquad\quad\forall\;I=1,2,\ldots,N,\\
    (e,0)\qquad&\mbox{ in }\;\Si_{ext}:=\bigcap_{I=1}^N\big(C_{\om_I,\th_I}(y_I)\cup B_1(y_I)\big)^c.
    \end{aligned}
    \right.
\end{align*}
As an immediate consequence of \eqref{gIkI'est}--\eqref{gIkI''est}, we obtain \eqref{gkformula} and \eqref{gkest} as stated.

By construction, for each $I\in\{1,2,\ldots,N\}$, we have that $B_{R_I}(\cb_I)\setminus B_{(1-2\de_I)R_I}(\cb_I)$ is a short-pulse annulus. Hence, a trapped surface will form in $D^+(B_{R_I}(\cb_I))$.
\paragraph{Step 2. Free of trapped surfaces.}
Fix $I\in\{1,\dots,N\}$ and consider the foliation
$S_r:=\partial B_r(\cb_I)$.
As in \cite[(6.14)]{ShenWan2}, we have
\begin{align}\label{in64R}
    \tr_g(\th_{S_r}-k_{S_r})\big|_p>0,
    \qquad \forall\,p\in S_r,\quad r\in(0,64R_I).
\end{align}
By construction,
\[
(g,k)=(g_I,k_I)
\quad\text{in }\;
(C_{\om_I,\frac12\th_I}(y_I)\cup B_{\frac12}(y_I))\cap B_{32R_I}^c(\cb_I),
\]
and by \eqref{gIkIbehavior} and \eqref{gIkI''est}, for $s\ge3$ and $\de<-\frac12$,
\[
g-e=O\big(r_I^{-\frac32-\de}\big),\qquad
k=O\big(r_I^{-\frac52-\de}\big)
\]
in $(C_{\om_I,\th_I}(y_I)\cup B_1(y_I))\cap B_{32R_I}^c(\cb_I)$.
Hence, for $r>32R_I$,
\[
\tr_{g_I}\th_{S_r}=\frac2r+O\big(r_I^{-\frac52-\de}\big),
\qquad
\tr_{g_I}k_{S_r}=O\big(r_I^{-\frac52-\de}\big),
\]
on $S_r\cap (C_{\om_I,\th_I}(y_I)\cup B_1(y_I))$, which yields, for
$-\frac32<\de<-\frac12$ and $R_I$ large,
\begin{align}\label{out32R}
    \tr_{g_I}(\th_{S_r}-k_{S_r})\big|_p>0,
    \qquad \forall\,r>32R_I,\;
    p\in S_r\cap (C_{\om_I,\th_I}(y_I)\cup B_1(y_I)).
\end{align}
Combining \eqref{in64R} and \eqref{out32R},
\begin{align}\label{Srcompare}
    \tr_g(\th_{S_r}-k_{S_r})\big|_p>0,
    \qquad \forall\,r>0,\;
    p\in S_r\cap (C_{\om_I,\th_I}(y_I)\cup B_1(y_I)).
\end{align}
Let $S$ be a compact embedded smooth $2$--surface with
$S\subset C_{\om_I,\th_I}(y_I)\cup B_1(y_I)$, and let $B_{r_S}(\cb_I)$ be the
smallest ball centered at $\cb_I$ containing $S$.
Then $S$ and $S_{r_S}$ are tangent at some $p$, and by mean curvature
comparison and \eqref{Srcompare},
\[
\tr_g(\th_S-k_S)\big|_p
\ge \tr_g(\th_{S_{r_S}}-k_{S_{r_S}})\big|_p>0,
\]
so $S$ is not trapped.

Finally, if a compact embedded smooth $2$--surface $S$ satisfies
$S\cap\Si_{ext}\neq\emptyset$, then for $p\in S\cap\Si_{ext}$,
since $(g,k)=(e,0)$ near $p$,
\begin{align}\label{extnotrapping}
    \tr_g(-\th_S-k_S)\,\tr_g(\th_S-k_S)\big|_p
    =-(\tr_e\th_S)^2\big|_p\le0,
\end{align}
contradicting the definition of a trapped surface.
This completes the proof of Theorem \ref{maintheorem}.
\end{proof}
\appendix
\section{Kerr-Schild initial data}\label{app:KS}
In this appendix, we record basic identities for metrics on $\RRR^{1+3}$ written in canonical coordinates $(t,x,y,z)$, of the Kerr-Schild form
\begin{align}\label{defgmunu}
\g_{\mu\nu}
=\etabf_{\mu\nu}+2H\ellbf_\mu\ellbf_\nu,
\end{align}
where $\etabf^{\mu\nu}\ellbf_\mu\ellbf_\nu=0$ and $H=H(x,y,z)$. Throughout, we normalize $\ellbf$ so that $\ellbf_0=1$. All statements apply, in particular, to the Kerr metric $\g_{m,a}$ written in Kerr-Schild form.
\begin{lem}\label{lem:KS-ADM}
Let $\g$ be a metric of the Kerr-Schild form \eqref{defgmunu} and let $\ellbf^\mu:=\etabf^{\mu\nu}\ellbf_\nu$. Then the following identities hold:
\begin{enumerate}
\item The metric components satisfy
\begin{align}
\begin{split}\label{eq:KS-inv}
\g^{\mu\nu}
&=\etabf^{\mu\nu}-2H\ellbf^\mu\ellbf^\nu,
\qquad
\g^{\mu\nu}\ellbf_\nu=\ellbf^\mu,
\qquad
\g_{\mu\nu}\ellbf^\mu\ellbf^\nu=0,
\\
\g_{ij}
&=\de_{ij}+2H\ellbf_i\ellbf_j,
\qquad
\g_{0i}=2H\ellbf_i,
\qquad
\g_{00}=-1+2H.
\end{split}
\end{align}
\item The lapse $\a$ and the shift $\b$ take the form
\begin{equation}\label{eq:KS-lapse-shift}
\a=(-\g^{00})^{-1/2}=(1+2H)^{-1/2},
\qquad
\b^i:=\a^2\g^{0i}=2\a^2H\ellbf^i,
\end{equation}
and satisfy
\begin{equation}\label{eq:KS-ADMids}
\b_i:=g_{ij}\b^j=\g_{0i},
\qquad
\g_{00}=-\a^2+g_{ij}\b^i\b^j.
\end{equation}
\item The future unit normal to the constant-time slices
$\Si_0=\{t=0\}$ is given by
\begin{equation}\label{eq:KS-normal}
\nbf_\mu=(-\a,0,0,0),
\qquad
\nbf^\mu=\g^{\mu\nu}\nbf_\nu=\a^{-1}(1,-\b^i).
\end{equation}
\end{enumerate}
\end{lem}
\begin{proof}
The identities in \eqref{eq:KS-inv} immediately follow from the fact that
$\ellbf$ is $\etabf$--null. For \eqref{eq:KS-lapse-shift}, we compute
\[
\g^{00}=\etabf^{00}-2H=-1-2H,
\qquad\quad
\g^{0i}=-2H\ellbf^0\ellbf^i=2H\ellbf^i,
\]
which yields
\[
\a=(1+2H)^{-1/2},
\qquad\quad
\b^i=2\a^2H\ellbf^i.
\]
Next, using $\ellbf_i\ellbf^i=1$ (by $\etabf$--nullness), we have
\[
g_{ij}\ellbf^j
=\big(\de_{ij}+2H\ellbf_i\ellbf_j\big)\ellbf^j
=\ellbf_i+2H\ellbf_i
=\a^{-2}\ellbf_i.
\]
Therefore,
\[
\b_i=g_{ij}\b^j=2\a^2H g_{ij}\ellbf^j=2H\ellbf_i=\g_{0i}.
\]
Then, using $\g_{0i}\g^{0i}=1-\g_{00}\g^{00}$, we compute
\begin{align*}
-\a^2+g_{ij}\b^i\b^j
&=-\a^2+\b_i\b^i
=-\a^2+\a^2\g_{0i}\g^{0i}
\\
&=-\a^2+\a^2\big(1-\g_{00}\g^{00}\big)
=\a^2(-\g_{00}\g^{00})
=\g_{00}.
\end{align*}
Finally, $\nbf_\mu=(-\a,0,0,0)$ is the unit conormal to $\Si_0$, and raising the
index yields
\[
\nbf^\mu=(-\a)\g^{\mu0}=(\a^{-1},-\a^{-1}\b^i).
\]
A direct computation gives
\[
\g_{\mu\nu}\nbf^\mu\nbf^\nu
=\nbf_\mu\nbf^\mu
=(-\a)\nbf^0
=-\a\a^{-1}
=-1.
\]
This concludes the proof.
\end{proof}
\begin{lem}\label{lem:KS-K}
The second fundamental form of $\Si_0=\{t=0\}$ induced by $\g$ is given by
\begin{equation}\label{eq:K-Lie}
k_{ij}
:=-\frac12(\Lie_\nbf\g)_{ij}
=\frac{1}{2\a}\big(
\b^l\pr_l g_{ij}
+g_{il}\pr_j\b^l
+g_{jl}\pr_i\b^l
-\pr_t g_{ij}
\big).
\end{equation}
\end{lem}
\begin{proof}
From \eqref{eq:KS-normal} we have
\[
\pr_t=\a\nbf+\b,\qquad\quad\b=\b^i\pr_i.
\]
Since $\g(\nbf,\pr_i)=0$, it follows that
\[
(\Lie_{\a\nbf}\g)_{ij}
=\a(\Lie_\nbf\g)_{ij}
-(\pr_i\a)\g(\nbf,\pr_j)
-(\pr_j\a)\g(\pr_i,\nbf)
=\a(\Lie_\nbf\g)_{ij}.
\]
Hence,
\[
\pr_t g_{ij}
=(\Lie_{\pr_t}\g)_{ij}
=(\Lie_{\a\nbf}\g)_{ij}+(\Lie_\b\g)_{ij}
=\a(\Lie_\nbf\g)_{ij}+(\Lie_\b\g)_{ij}.
\]
Moreover,
\begin{align*}
(\Lie_\b g)_{ij}
&=\b^l\pr_l(g_{ij})
-g(\Lie_\b\pr_i,\pr_j)
-g(\pr_i,\Lie_\b\pr_j)
\\
&=\b^l\pr_l(g_{ij})
-g([\b^k\pr_k,\pr_i],\pr_j)
-g(\pr_i,[\b^k\pr_k,\pr_j])
\\
&=\b^l\pr_l(g_{ij})
+g_{jk}\pr_i(\b^k)
+g_{ik}\pr_j(\b^k).
\end{align*}
Combining the above identities yields \eqref{eq:K-Lie}.
This concludes the proof.
\end{proof}
\section{Computation of Kerr initial data in Kerr-Schild form}
In this Appendix, we provide the detailed proofs of Lemmas \ref{preidentities}, \ref{expressionk} and \ref{moreidentities}. We will make use of the shorthanded $\O^p_q$--notation introduced in Definition \ref{KerrO} to denote the lower order terms.
\subsection{Proof of Lemma \ref{preidentities}}\label{proof-preidentities}
From \eqref{dfrt}, we have
\[
\rt^4-(r^2-a^2)\rt^2-a^2z^2=0.
\]
Hence, we obtain
\begin{align*}
    \rt^2=r^2+\O_0^2\qquad \mbox{ and }\qquad\rt=r(1+\O_2^2).
\end{align*}
Therefore,
\[
H=\frac{m\rt^3}{\rt^4+a^2z^2}=\frac{m}{r}(1+\O_2^2),
\qquad
\pr_iH=-\frac{mx_i}{r^3}+\O_4^3,
\qquad
\a=(1+2H)^{-1/2}=1-\frac{m}{r}+\O_2^2.
\]
Moreover, from \eqref{KerrSchildcomponents}, we compute
\begin{align}
\begin{split}\label{ellbfl}
\ellbf_1&=\frac{\rt x+ay}{\rt^2+a^2}=\left(\frac{x}{r}+\frac{ay}{r^2}\right)(1+\O_2^2),\qquad
\ellbf_2=\frac{\rt y-ax}{\rt^2+a^2}=\left(\frac{y}{r}-\frac{ax}{r^2}\right)(1+\O_2^2),\\
\ellbf_3&=\frac{z}{\rt}=\frac{z}{r}(1+\O_2^2).
\end{split}
\end{align}
These immediately imply
\begin{align}\label{computexell}
\ellbf_j \frac{x_j}{r}=\frac{x}{r}\left(\frac{x}{r}+\frac{ay}{r^2}\right)(1+\O_2^2)+\frac{y}{r}\left(\frac{y}{r}-\frac{ax}{r^2}\right)(1+\O_2^2)+\frac{z^2}{r^2}(1+\O_2^2)=1+\O_2^2.
\end{align}
We also have
\begin{align*}
\ellbf_i\pr_iH&=\left(\frac{x}{r}+\frac{ay}{r^2}\right)(1+\O_2^2)\left(-\frac{mx}{r^3}+\O_4^3\right)+\left(\frac{y}{r}-\frac{ax}{r^2}\right)(1+\O_2^2)\left(-\frac{my}{r^3}+\O_4^3\right)\\
&+\frac{z}{r}(1+\O_2^2)\left(-\frac{mz}{r^3}+\O_4^3\right)\\
&=-\frac{m}{r^2}+\O_4^3,
\end{align*}
and
\begin{align*}
\frac{x_j}{r}\pr_jH=\frac{x}{r}\left(-\frac{mx}{r^3}+\O_4^3\right)+\frac{y}{r}\left(-\frac{my}{r^3}+\O_4^3\right)+\frac{z}{r}\left(-\frac{mz}{r^3}+\O_4^3\right)=-\frac{m}{r^2}+\O_4^3.
\end{align*}
Next, using \eqref{ellbfl}, we compute
\begin{align*}
    \pr_i\ellbf_i&=\pr_x\left[\left(\frac{x}{r}+\frac{ay}{r^2}\right)(1+\O_2^2)\right]+\pr_y\left[\left(\frac{y}{r}-\frac{ax}{r^2}\right)(1+\O_2^2)\right]+\pr_z\left[\frac{z}{r}(1+\O_2^2)\right]\\
    &=\left(\frac{1}{r}-\frac{x^2}{r^3}-\frac{2ayx}{r^4}\right)(1+\O_2^2)+\left(\frac{x}{r}+\frac{ay}{r^2}\right)\O_3^2+\left(\frac{1}{r}-\frac{y^2}{r^3}+\frac{2axy}{r^4}\right)(1+\O_2^2)\\
    &+\left(\frac{y}{r}-\frac{ax}{r^2}\right)\O_3^2+\left(\frac{1}{r}-\frac{z^2}{r^3}\right)(1+\O_2^2)+\frac{z}{r}\O_3^2\\
    &=\frac{2}{r}+\O_3^2.
\end{align*}
Finally, we have from \eqref{computexell}
\begin{align*}
\ellbf_i\pr_i\ellbf_j\frac{x_j}{r}=\frac{1}{r}\ellbf_i\pr_i(\ellbf_j x_j)-\frac{1}{r}\ellbf_i\ellbf_j\de_{ij}=\frac{\ellbf_i}{r}\pr_i(r+\O_1^2)-\frac{1}{r}=\frac{\ellbf_i}{r}\left(\frac{x_i}{r}+\O_2^2\right)-\frac{1}{r}=\O_3^2.
\end{align*}
This concludes the proof of Lemma \ref{preidentities}.
\subsection{Proof of Lemma \ref{expressionk}}\label{proof-expressionk}
Since the Kerr metric is stationary, we have $\pr_t g_{ij}=0$ throughout the computation. We first compute
\begin{align*}
\pr_i(\a^2H)&=\pr_i\left(\frac{H}{1+2H}\right)=\frac{1}{(1+2H)^2}\pr_iH=\a^4\pr_iH,\\
\ellbf_l\pr_i\ellbf_l&=\frac{1}{2}\pr_i(\ellbf_l\ellbf_l)=0,
\end{align*}
since $\sum_{l=1}^3 \ellbf_l^2 = 1$. Then, we have from \eqref{eq:KS-lapse-shift} and \eqref{eq:K-Lie}
\begin{align*}
\a k_{ij}&=\a^2H\ellbf_l\pr_l (g_{ij})+g_{il}\pr_j(\a^2H\ellbf_l)+g_{jl}\pr_i(\a^2H\ellbf_l)\\
&=\a^2H\ellbf_l\pr_l(2H\ellbf_i\ellbf_j)+\pr_j(\a^2H\ellbf_i)+\pr_i(\a^2H\ellbf_j) +2H\ellbf_i\ellbf_l\pr_j(\a^2H\ellbf_l)+2H\ellbf_j\ellbf_l\pr_i(\a^2H\ellbf_l)\\
&=\a^2H\ellbf_l\pr_l(2H\ellbf_i\ellbf_j)+\a^2H(\pr_j\ellbf_i+\pr_i\ellbf_j)+\a^{-2}\ellbf_i\pr_j(\a^2H)+\a^{-2}\ellbf_j\pr_i(\a^2H)\\
&=\a^2H\ellbf_l\pr_l(2H\ellbf_i\ellbf_j)+\a^2H(\pr_j\ellbf_i+\pr_i\ellbf_j) +\a^2(\ellbf_i\pr_jH+\ellbf_j\pr_iH)\\
&=\a^2H\ellbf_l\pr_l(2H\ellbf_i\ellbf_j)+\a^2(\pr_j(H\ellbf_i)+\pr_i(H\ellbf_j)),
\end{align*}
which implies
\begin{equation}
\a^{-1}k_{ij}=2H\ellbf_l\pr_l(H\ellbf_i\ellbf_j)+ \pr_j(H\ellbf_i)+\pr_i(H\ellbf_j).
\end{equation}
Taking the trace, one obtains
\begin{equation}
\a^{-1}\tr_e k=2H\ellbf_l\pr_lH+2 \pr_l(H\ellbf_l)=2(1+H)\ellbf_l\pr_lH+2H\pr_l\ellbf_l.
\end{equation}
This concludes the proof of Lemma \ref{expressionk}.
\subsection{Proof of Lemma \ref{moreidentities}}\label{proof-moreidentities}
From \eqref{ellbfl}, we compute
\begin{align*}
    \ellbf_i\pr_i\ellbf_1&=\left(\frac{x}{r}+\frac{ay}{r^2}\right)\pr_x\left[\left(\frac{x}{r}+\frac{ay}{r^2}\right)(1+\O_2^2)\right](1+\O_2^2) +\left(\frac{y}{r}-\frac{ax}{r^2}\right)\pr_y\left[\left(\frac{x}{r}+\frac{ay}{r^2}\right)(1+\O_2^2)\right](1+\O_2^2)\\
    &+\frac{z}{r}\pr_z\left[\left(\frac{x}{r}+\frac{ay}{r^2}\right)(1+\O_2^2)\right](1+\O_2^2)\\
    &=\left(\frac{x}{r}+\frac{ay}{r^2}\right)\left(\frac{1}{r}-\frac{x^2}{r^3}-\frac{2ayx}{r^4}\right)+\left(\frac{y}{r}-\frac{ax}{r^2}\right)\left(-\frac{xy}{r^3}+\frac{a}{r^2}-\frac{2ay^2}{r^4}\right)\\
    &+\frac{z}{r}\left(-\frac{xz}{r^3}-\frac{2ayz}{r^4}\right)+\O_3^2\\
    &=\O_3^2,
\end{align*}
and
\begin{align*}
    \frac{x_i}{r}\pr_i\ellbf_1&=\frac{x}{r}\pr_x\left[\left(\frac{x}{r}+\frac{ay}{r^2}\right)(1+\O_2^2)\right]+\frac{y}{r}\pr_y\left[\left(\frac{x}{r}+\frac{ay}{r^2}\right)(1+\O_2^2)\right]+\frac{z}{r} \pr_z\left[\left(\frac{x}{r}+\frac{ay}{r^2}\right)(1+\O_2^2)\right]\\
    &=\frac{x}{r}\left(\frac{1}{r}-\frac{x^2}{r^3}-\frac{2ayx}{r^4}\right)+\frac{y}{r}\left(-\frac{xy}{r^3}+\frac{a}{r^2}-\frac{2ay^2}{r^4}\right)+\frac{z}{r}\left(-\frac{xz}{r^3}-\frac{2ayz}{r^4}\right)+\O_3^2\\
    &=-\frac{a y}{r^3}+\O_3^2.
\end{align*}
Similarly, we have
\begin{align*}
    \ellbf_i\pr_i\ellbf_2=\O_3^2,\qquad\qquad
    \frac{x_i}{r}\pr_i\ellbf_2
    =\frac{a x}{r^3}+\O_3^2.
\end{align*}
Finally, we compute
\begin{align*}
    \ellbf_i\pr_i\ellbf_3&=\left(\frac{x}{r}+\frac{ay}{r^2}\right)\pr_x\left[\frac{z}{r}(1+\O_2^2)\right](1+\O_2^2)+\left(\frac{y}{r}-\frac{ax}{r^2}\right)\pr_y\left[\frac{z}{r}(1+\O_2^2)\right](1+\O_2^2)\\
    &+\frac{z}{r}\pr_z\left[\frac{z}{r}(1+\O_2^2)\right](1+\O_2^2)\\
    &=\left(\frac{x}{r}+\frac{ay}{r^2}\right)\left(-\frac{zx}{r^3}\right)+\left(\frac{y}{r}-\frac{ax}{r^2}\right)\left(-\frac{zy}{r^3}\right)+\frac{z}{r}\left(\frac{1}{r}-\frac{z^2}{r^3}\right)+\O_3^2\\
    &=\O_3^2,
\end{align*}
and
\begin{align*}
    \frac{x_i}{r}\pr_i\ellbf_3&=\frac{x}{r}\pr_x\left[\frac{z}{r}(1+\O_2^2)\right]+\frac{y}{r}\pr_y\left[\frac{z}{r}(1+\O_2^2)\right]+\frac{z}{r}\pr_z\left[\frac{z}{r}(1+\O_2^2)\right]\\
    &=\frac{x}{r}\left(-\frac{xz}{r^3}\right)+\frac{y}{r}\left(-\frac{yz}{r^3}\right)+\frac{z}{r}\left(\frac{1}{r}-\frac{z^2}{r^3}\right)+\O_3^2\\
    &=\O_3^2.
\end{align*}
This concludes the proof of Lemma \ref{moreidentities}.

\end{document}